\documentclass[11pt,reqno]{article}

\usepackage[all]{xy}           
\usepackage{amssymb,amsmath}           
\usepackage{amsthm}
\usepackage{hyperref}
\usepackage{eucal}
\usepackage{nicefrac}
\usepackage{color}

\parskip=\medskipamount
\newtheorem{definition}{\noindent \noindent {\bf
Definition}}[section]
\newtheorem{lem}{{\bf Lemma}}[section]

\newtheorem{prop}{{\bf Proposition}}[section]

\def\r{\ensuremath{\mathbb{R}}}
\def\rk{{\mathbb R}^{k}}

\def\tkm{T^1_kM}
\def\tkq{T^1_kQ}

\def\d{{\rm d}}

\def\derpar#1#2{\ds\frac{\partial{#1}}{\partial{#2}}}

\def\bea{\begin{eqnarray}}
\def\eea{\end{eqnarray}}
\def\beq{\begin{equation}}
\def\eeq{\end{equation}}
\def\beann{\begin{eqnarray*}}
\def\eeann{\end{eqnarray*}}

\newcommand{\ds}{\displaystyle}

\renewcommand{\neq}{=\hspace{-3.5mm}/\hspace{2mm}}

\def\fpd#1#2{{\frac{\partial #1}{\partial #2}}}

\def\spd#1#2#3{{\frac{\partial^2 #1}
{\partial #2\partial #3}}}

\def\vectorfields#1{{\mathfrak X}(#1)}
\def\cinfty#1{C^{\scriptscriptstyle\infty}(#1)}
\def\g{\mathfrak{g}}

 \definecolor{ochre}{rgb}{0.8, 0.47, 0.13}

\title{Symmetry reduction, integrability and reconstruction in $k$-symplectic field theory}

\author{L.\  B\'{u}a${}^a$, T.\ Mestdag${}^b$ and M.\ Salgado${}^a$
\\[2mm]${}^a$ Departamento de Xeometr{\'i}a e Topolox{\'i}a,\\
  Universidade de Santiago de Compostela, Spain.
\\[2mm]${}^b$ Department of Mathematics, Ghent University,\\ Krijgslaan 281, B--9000 Gent, Belgium}
\date{}
\begin{document}

\maketitle


\date{}

\begin{abstract}
We investigate the reduction process of a $k$-symplectic field theory whose Lagrangian is invariant under a symmetry group. We give explicit coordinate expressions of the resulting reduced partial differential equations, the so-called Lagrange-Poincar\'e field equations. We discuss two issues about reconstructing a solution from a given solution of the reduced equations. The first one is an interpretation of the integrability conditions, in terms of the curvatures of some connections. The second includes the introduction of the concept of a $k$-connection to provide a reconstruction method. We show that an invariant Lagrangian, under suitable regularity conditions, defines a `mechanical' $k$-connection.
\end{abstract}








\section{Introduction}\label{section 1}


The Lagrangian equations of a first-order field theory are a set of second-order partial differential equations in the unknown fields $\phi^A(t)$, depending on $k$ parameters $t^\alpha$. For a Lagrangian $L(q^A,u^A_\alpha)$, they are of the form
\begin{equation}\label{lfield}
 \displaystyle\frac{\partial^2 L}{\partial q^B \partial u^A_\alpha}  \frac{\displaystyle\partial\phi^B} {\displaystyle\partial
t^\alpha}  +
 \ds\frac{\partial^2 L}{\partial u_\beta^B \partial u^A_\alpha  }  \frac{\displaystyle\partial^2\phi^B} {\displaystyle\partial
t^\alpha \partial t^\beta}   =  \ds\frac{\partial  L}{\partial q^A},
  \end{equation}
  with $(q^A=\phi^A(t),u^A_\alpha = \partial \phi^A/t^\alpha(t))$. In the literature, there exist many geometric models that describe classical Lagrangian field equations. Just to name a few, we mention the polysymplectic \cite{sarda1,Kana}, the $n$-symplectic \cite{no1},  the
  $k$-cosymplectic \cite{mod2},  the
multisymplectic  \cite{CCI91,bar1,GIM1,KijTul} and the jet \cite{Krupka,Saunders} formalisms. The main differences between all these models depend on e.g.\ the choice one makes for the geometric and the differentiable structure of both the space of parameters $t^\alpha$ (such as e.g. spacetime) and the space of fields $\phi^A$. The model we will use in this paper is the one of $k$-symplectic field theory, as developed in e.g.\ the papers \cite{BBS,gunther,fam,rsv07}. The space where the derivatives of the fields, $\partial\phi^a/\partial t^\alpha$, live is identified in this setting with the so-called tangent bundle of $k^1$-velocities $T^1_kQ$. In many ways, one may think of $k$-symplectic field theory as the model that resembles the closest the standard symplectic formalism of both Lagrangian and Hamiltonian mechanics on a tangent and a cotangent bundle, respectively. It characterises the (regular) field theory in terms of a certain class of so-called `$k$-vector fields' on $T^1_kQ$, which are literally collections of $k$ individual vector fields.

In the last few years there has been an increasing interest in field theories with symmetry, and in their reduction (see e.g.\ \cite{CGR,CRS,Ellis,Ellis2,marrero,JV} and the references therein). Depending on the nature of the space of fields, the reduced PDEs are often referred to as the `Lagrange-Poincar\'e field equations' or the `Euler-Poincar\'e field equations'. The general idea behind symmetry reduction is that, when a dynamical system (be it a set of ODEs or PDEs) is invariant under the action of a symmetry Lie group, the system can be reduced to one in fewer variables which is presumably easier to solve. The second step in the process is to reconstruct a solution of the original dynamical system from a given solution of the reduced system.


  The main goal of the paper is to show how both the reduction and reconstruction process works in the context of $k$-symplectic field theories. The method that has been followed the most up to now in the literature (for different geometric models of Lagrangian field theories), depends on a reduction of the variational principle that generates the Lagrangian equations.  By contrast, we will show that the $k$-symplectic model is ideal to follow a somewhat different procedure, which is similar to the one that has been used in the paper \cite{MC} for Lagrangian systems with symmetry. In our formulation of reduction below, we will bring the $k$-vector fields to the front, rather than the (unreduced or reduced) PDEs they produce.

 After some preliminaries (in Section~\ref{sec2}) we discuss in Section~\ref{sec5} some results about the integrability conditions of the PDEs that can be associated to an arbitrary invariant $k$-vector field ${\mathbf X}$ on a manifold $M$. Under the assumption that the reduced equations on $M/G$ are integrable, we will give an interpretation of the remaining integrability conditions in terms of the  curvature of some connection $\omega^{\breve\phi,{\mathbf X}}$.

  We then specify to the case where $M=T^1_kQ$, and the dynamics to those given by a Lagrangian $k$-vector field. In Section~\ref{sec3} we present a new formulation of the Lagrangian  $k$-vector fields on $T^1_kQ$ in terms of a non-standard local frame of vector fields on $Q$. In the presence of a Lagrangian with a symmetry group $G$, we identify in Section~\ref{sec4} the action under which the Lagrangian $k$-vector fields are invariant, and we show that they can be reduced to $k$-vector fields on the reduced space $(T^1_kQ)/G$. We end Section~\ref{sec4} with a computation of the coordinate expressions of these vector fields and their associated PDEs (which represent the Lagrange-Poincar\'e PDEs in this context).

At the end of Section~\ref{sec4} we turn back to our interpretation of the integrability conditions, for the case of a Lagrangian $k$-vector field and we make the link, in our setting, to some results about `reconstruction' that have appeared in the paper \cite{Ellis}. There, a big role is played by two connections ${\mathcal A}^\rho$ and   ${\mathcal A}^{\bar\sigma}$. We will show how these connections (i.e.\ their analogues, when translated to our setting) appear in our discussion about integrability, by decomposing the connection $\omega^{\breve\phi,{\mathbf X}}$ into two parts.

The integrability conditions only guarantee that a solution may be reconstructed, but they do not tell one how to do so. In Section~\ref{sec6} we discuss, first for a $k$-vector field $\mathbf X$ on $M$, a reconstruction method that allows one to re-assemble the solution, from a given solution of the reduced equations and from a map that takes values in the symmetry Lie group. This part of the problem involves the introduction of a new concept, that of a principal $k$-connection on the principal bundle $M\to M/G$. It is an appropriate generalization, to the level of $k$-tangent bundles, of the notion of a principal
connection. We end Section~\ref{sec6} by showing that, on $M=T^1_kQ$,  such a connection is naturally available for a Lagrangian field theory with symmetry (up to a certain regularity condition  on the Lagrangian). Since it resembles the so-called mechanical connection which appears in the context of a Lagrangian system whose kinetic energy is associated to a Riemannian metric (see e.g.\ \cite{MC} for a discussion on this topic), we have kept that name also for the case of field theories. We end the paper with an application of our results to the context of harmonic maps.

\section{Integrability of a $k$-vector field}  \label{sec2}

In this section we recall the concept of a $k$-vector field, and of an integral section of a $k$-vector field. Parts of this section can be found in more detail in the papers \cite{BBS,fam,rsv07}. We finish the section with a useful integrability criterion for a $k$-vector field in terms of an associated connection.

\subsection{Connections and curvature}

In what follows we will often use non-linear connections, on many bundles. To set notations, let us recall briefly their definition. Let $p: E \to B$ be a fibre bundle. For $e\in E$, the vertical space $V_eE$ at $e$ is given by the kernel of $T_ep: T_eE \to T_{p(e)}B$. It gives rise to the so-called vertical distribution $VE=\{V_eE  | e\in E\}$. We can put this in a short exact sequence of vector bundles over $M$,
\begin{equation} \label{seq}
0 \to VE \to TE \to E\times_B TB \to 0,
\end{equation}
where the middle arrow $j: TE \to E\times_B TB$ is given by $v_e \mapsto (e, Tp(v_e))$. A connection on $p$ is either given by a right splitting $\gamma: E\times_B TE \to TE$ (i.e.\ a linear  map satisfying $j \circ \gamma = id$), or by the corresponding left splitting $\omega = id -\gamma    \circ j: TE \to VE\subset TE$.

The above short exact sequence naturally extends to the level of sections of the corresponding bundles over $M$,
\[
0 \to Sec(VE) \to \vectorfields{E} \to Sec(E\times_B TB) \to 0.
\]
A splitting of (\ref{seq}) induces a splitting of the second sequence. When we interpret $\omega: \vectorfields{E} \to \vectorfields{E}$ as a (1,1) tensor field on $E$, we will call it the {\sl connection form}, or the {\sl vertical projection}. The map $h:id-\omega$ is the {\sl horizontal projection} of the connection. Since vector fields $T$ on $B$ can be thought of as basic sections in $Sec(E\times_B TB)$, we may define the {\sl horizontal lift of $T$} as the vector field $T^h$ of $E$, given by $T^h(e) = \gamma(e,T(\pi(e)))$.

The {\sl curvature} of the connection is the (1,2) tensor field on $E$, given by $(X,Y) \mapsto -\omega([hX,hY])$, for two vector fields $X,Y\in\vectorfields{E}$. In what follows, however, we will also often use the word `curvature' for the restriction of that map to two horizontal lifts and use the notation
\[
K(T,S) = - \omega([T^h,S^h])  \in\vectorfields{E}
\]
when  $T,S\in\vectorfields{B}$.

\subsection{The tangent bundle of $k^1$-velocities}

Let $\tau_M\colon TM \to M$ be the tangent bundle of a differentiable manifold $M$. We will use the notation $T^1_kM$ for the Whitney sum
$TM\oplus\stackrel{k}{\dots}\oplus TM$ of $k$ copies of $TM$ and $\tau^1_M$
 for the corresponding projection $\tau^1_M \colon T^1_kM\to M$ which maps $({u_1},\ldots ,
{u_k})$ onto the point $m\in M$ on which all $u_\alpha$'s $\tau_M$-project.
$T^1_kM$ can be identified with the manifold $J^1_0({\bf R}^k,M)$ of $k^1$-velocities of
$M$. These are $1$-jets of maps from $\rk$ to $M$ with source at $0\in {\bf R}^k$.
For this reason the manifold $T^1_kM$ is called {\sl the tangent bundle of
$k^1$-velocities of $M$}.

In what follows, we will denote coordinates on $\r^k$ by  $ (t^\alpha) = (t^1,\ldots, t^k)$ and use bold face letters ${\mathbf u}$ to denote
elements $({u_1},\ldots , {u_k})$ in $T^1_kM$. If $(x^I)$ (with $I=1\ldots \dim M$) are local coordinates on $U \subset M$ then the
induced local coordinates   $(x^I,
u^I)$ on $TU=\tau_M^{-1}(U)$ are given by
$$
x^I( u_m)=x^I(m),\qquad u^I(u_m)=u_m(x^I), \quad u_m\in T_mM.
$$
These naturally induce coordinates $(x^I ,
u^I_\alpha)$ (with $I=1\ldots  \dim M;\, \alpha=1\ldots k$) for a point ${\mathbf u}$ in
$T^1_kU=(\tau_M^1)^{-1}(U)$,
such that  $u^I_\alpha$ are the components of the $\alpha$'th vector ${u_\alpha}$ of ${\mathbf u}$ along the natural basis of $T_mM$.

Let $\varphi\colon M\to N$ be a differentiable map. In what follows, we will  make use of the {\sl canonical prolongation of $\varphi$}, which is  the induced
map $T^1_k\varphi:T^1_kM \to  T^1_k N$  defined by
\[
  T^1_k\varphi({\mathbf v})=
(T_m\varphi(v_1),\ldots,T_m\varphi(v_k)) \  .
\]
The {\sl first prolongation $\psi^{(1)}$ of a map $\psi: \r^k \to M$} is the map $\r^k\to T^1_kM$, defined by
\[
\psi^{(1)}(t)=
 \left(T_t\psi\left(\derpar{}{t^1}\Big\vert_t\right),\ldots,
T_t\psi\left(\derpar{}{t^k}\Big\vert_t\right)\right) \, .
\]
In local coordinates, we have
\begin{equation}\label{localfi11}
\psi^{(1)}(t)=\left( \psi^I (t), \frac{\partial\psi^I}{\partial t^\alpha} (t)\right), \qquad  1\leq \alpha\leq k\, ,\, 1\leq I\leq \dim M \, .
\end{equation}

\begin{definition}
 A  {\rm $k$-vector} field on $M$ is a section ${\bf X}: M \to T^1_kM$
of the vector bundle $\tau^1_M:T^1_kM\to M$.
\end{definition}

Given that $T^{1}_{k}M$ is the Whitney sum  of $k$ copies of $TM$, by projecting a $k$-vector field ${\bf X}$  onto
every factor, we see that it consists of a family of $k$ vector
fields $X_\alpha = \tau_\alpha \circ{\bf X}$ on $M$. Here,
$\tau_\alpha\colon T^1_kM \rightarrow TM$ stands for the canonical projection
on the $\alpha^{th}$-copy of $TM$ in $T^1_kM$. We will denote the set of $k$-vector fields on $M$ by  $\mathfrak{X}^k(M)$. It is a $\cinfty{M}$-module.

\begin{definition}
\label{integsect} An integral section  of a $k$-vector
field ${\mathbf X}$, passing through a point
$m\in M$, is a map $\psi: U_0\subset \r^k \rightarrow M$,
defined on some neighbourhood  $U_0$ of $0\in \r^k$,  such that $\psi(0)=m$ and ${\mathbf X}\circ\psi=\psi^{(1)}$, where  $\psi^{(1)}$ is the first
prolongation of $\psi$.

 A $k$-vector field ${\mathbf X}$ on $M$ is integrable if there exists an integral section passing through every point of $M$.
\end{definition}

An integral section $\psi$ of the $k$-vector field ${\mathbf X}$, consisting of the vector fields $X_{\alpha} = X^I_\alpha \partial /\partial x^I$,  satisfies
\begin{equation}\label{condInteg0}
T_t\psi \left(\frac{\partial}{\partial
t^\alpha}\Big\vert_t\right)=X_{\alpha}(\psi (t)),
\end{equation}

and therefore, in view of (\ref{localfi11}), we get in local coordinates
\begin{equation}\label{condInteg}
\fpd{\psi^I}{t^\alpha} = X^I_\alpha \circ \psi, \quad 1\leq I \leq \dim M, \,\, 1\leq \alpha\leq k.
\end{equation}

As is the case for any vector bundle, by considering products of the bundle with its dual, we may consider forms and tensor fields on it. For example, we will speak of a {\em (1,1) $k$-tensor field on $M$} when we mean a (1,1) tensor field on $\tau^1_k$, i.e.\ is a $\cinfty{M}$-linear  map  $\mathfrak{X}^k(M) \to\mathfrak{X}^k(M)$. Locally, we can write for ${\mathbf X} = (X_\alpha)$ that ${\mathbf A}({\mathbf X}) = {\mathbf Y}$, with $Y^J_\beta = A^{J\alpha}_{\beta I} X^I_\alpha$. As a special case, one may consider a (1,1) tensor field $A$ on $M$, and extend it to a (1,1) $k$-tensor field, by putting $Y_\beta=A(X_\beta)$. Then $A^{J\alpha}_{\beta I} = A^J_I \delta^\beta_\alpha$. In an analogous terminology, we will speak of $(r,s)$ $k$-tensor fields on $M$.

\begin{definition}
Let  $X$ be a vector field on $M$. The Lie derivative ${\mathcal L}_X$ of a $k$-vector field $\mathbf Y$ on $M$ is the $k$-vector field  ${\mathcal L}_X \mathbf Y$ on $M$ whose $\alpha$th component is given by the vector field $[X, Y_\alpha]$ on $M$.
\end{definition}
 An equivalent formulation that makes use of the flow $\phi_t$ of the vector field $X$ is then
\[
{\mathcal L}_X {\mathbf Y}(m) =
\lim_{t\mapsto 0} \frac{T^1_k\phi_t({\mathbf Y}(m)) - {\mathbf Y}(m) }{t}.
\]
The corresponding Leibnitz-property is then ${\mathcal L}_X (f{\mathbf Y}) = X(f){\mathbf Y} +f {\mathcal L}_X \mathbf Y$. We can easily extend the Lie derivative to $k$-tensor fields. In particular, we have
\begin{equation}\label{lieder}
({\mathcal L}_X{\mathbf A})(\mathbf Y) = {\mathcal L}_X({\mathbf A}(\mathbf Y)) - {\mathbf A}({\mathcal L}_X \mathbf Y)
\end{equation}
for a (1,1) $k$-tensor field.

\subsection{The connection associated to a $k$-vector field}

An arbitrary section of the trivial bundle $\pi\colon\rk\times M \to \rk$ can be written as $(Id_{\rk},\phi):\rk \to \rk\times M$. Elements of the first jet bundle $J^1\pi$
can then be identified with couples in $\rk\times T^1_kM$, as follows:
$$j^1_t(Id_{\rk},\phi) \equiv \left(t,(\ldots, \phi_*(t)(\frac{\partial}{\partial t^\alpha}\Big\vert_{t}),
\ldots)\right)=(t,\phi^{(1)}(t)).$$

Each $k$-vector field ${\bf X}=(X_1,\ldots, X_k)$ defines a special type of {\sl jet field on $\pi$}, given by
\[
(Id_{\rk}, {\bf X})\colon \rk\times M \to \rk\times T^1_kM,
\]
 that is: a section of the canonical projection $\pi_{1,0}:J^1\pi \equiv \rk\times T^1_kM \to \rk\times M$.
  In local fibered coordinates the section is given by  $$(Id_{\rk}, {\bf X})(t^\alpha,x^I)=(t^\alpha,x^I,X^I_\alpha(x)).$$

A section  $\bar{\psi}=(Id_{\rk},\psi)$ of $\pi$ is called an integral section of the jet field $(Id_{\rk}, {\bf X})$ if $j^1\bar{\psi}=(Id_{\rk}, {\bf X})\circ\bar{\psi}$. This means, locally,  that
$\psi\colon \rk \to M$ must satisfy the equations (\ref{condInteg}), or that it must be an integral section of ${\bf X}$.
The $k$-vector field ${\bf X}$ is therefore integrable if and only if it is  associated jet field $(Id_{\rk},{\bf X})$ of $\pi$ is.

It is well-known that jet fields may be interpreted as connections (see \cite{Saunders} for details). In particular, the jet field $(Id_{\rk}, {\bf X})$  of a $k$-vector field can be identified with a connection on
the trivial bundle $\rk\times M$,  i.e.\ with a splitting of the short exact sequence
\[
0 \to V(\r^k\times M)\equiv\r^k \times TM   \to T(\r^k\times M) \to (\r^k\times M)\times_{\r^k\times M}T\r^k \equiv M \times T\r^k \to 0
\]
of vector bundles over $\rk\times M$. The right splitting  $\gamma^{\bf X}\colon  M \times T\r^k \to T(\rk \times M)$ of the {\sl connection associated to  ${\bf X}$} is given by
 \begin{equation}\label{locgam0}\gamma^{\bf X}  (m,T^\alpha \derpar{}{t^\alpha}\Big\vert_{t})= T^\alpha \left(
 \derpar{}{t^\alpha}\Big\vert_{(t,m)} +X^I_\alpha(m)\derpar{}{x^I}\Big\vert_{(t,m)}\right)
 \end{equation}
 and there is a similar expression for its horizontal lift.
The left splitting $\omega^{\bf X}\colon  T(\rk\times M) \to \r^k \times TM$  is, under the identifications made, given by
 \begin{equation}\label{locomega0}\omega^{\bf X}
 (T^\alpha \derpar{}{t^\alpha}\Big\vert_{(t,m)}+ Y^I   \derpar{}{x^I}\Big\vert_{(t,m)}         )=
\left(t ,(Y^I-X^I_\alpha\, T^\alpha)  \derpar{}{x^I}\Big\vert_{m}\right).
 \end{equation}

One easily verifies (see e.g.\ Proposition~4.6.10 in \cite{Saunders}) that the jet field $(Id_{\rk}, {\bf X})$
 is integrable, or equivalently: that the $k$-vector field ${\bf X}$ is integrable, if and only if the curvature
$K^{\bf X}$ of the associated connection vanishes. This is equivalent with $[X_\alpha,X_\beta]=0$, or
\begin{equation}\label{condInteg1}
X^I_\alpha\fpd{X^J_\beta}{x^I} - X^I_\beta\fpd{X^J_\alpha}{x^I} = 0,
\end{equation}
for all $x\in M$.

\section{Integrability of an invariant $k$-vector field}\label{sec5}

In this section we study the above integrability criterion for the case of a $G$-invariant $k$-vector field. We examine its relationship to the integrability criterion of its reduced $k$-vector field.

 \subsection{Invariant $k$-vector fields} \label{sec41}

 Let $\Phi : G {\bf \times} M \to M$ be a free and proper action of a connected Lie group $G$ on $M$. Then, the projection $\pi_M:M\to M/G$ on the set of
equivalence classes defines a principal bundle structure on $M$. A  vector field $W$ on $M$ is said to be invariant if
$$ T_m\Phi_g(W(m))=W(\Phi_g(m))\, .
$$
In that case, the relation
 \begin{equation}\label{defbreve}
 \breve{W} \circ \pi_M=T\pi_M \circ W
 \end{equation}
 uniquely defines a {\sl reduced vector field $\breve{W}$ on $M/G$}.

Likewise, if $F:M \to \r$ is an invariant function on $M$ it can be reduced to a function $f:M/G \to \r$ with
$
f \circ \pi_M\, = \, F$.
We also have that
\begin{equation}\label{wf}
W(F) \, =\, W(f \circ \pi_M)\,  = \, \breve{W}(f) \circ \pi_M,
\end{equation}
that is, $\breve{W}(f)$ is the reduced function on $M /G$ of the invariant function $W(F)$ on $M$.

   We will
denote by $\Phi^{\tkm}:G\times T^1_kM \to T^1_kM$ the $k$-tangent action, given by $\Phi^{\tkm} (g,{\mathbf v}) = T^1_k\Phi_g({\mathbf v})$, or
\[
\Phi^{\tkm} (g,v_{1}, \ldots, v_{k}) =
\left(T_m\Phi_g (v_{1}) , \ldots , T_m\Phi_g
(v_{k})\right),
\]
where $m=\tau^1_M({\mathbf v})$ and $g \in G$. The action $\Phi^{\tkm}:G\times \tkm \to \tkm$ is also
free and proper and, therefore, $\pi_{\tkm}:\tkm \to (\tkm)/G$ is a
principal bundle too.

\begin{definition}
A $k$-vector field ${\mathbf X}$ on $M$ is $G$-invariant if $\Phi^{T^1_kM}_g \circ {\mathbf X} = {\mathbf X} \circ\Phi_g$.
\end{definition}
Thus, a $k$-vector field ${\mathbf X}$ on $M$ is $G$-invariant if
$$T_m\Phi_g( X_\alpha(m) )=   X_\alpha(\Phi_g(m)) \,\quad m\in M ,\,\,  1\leq \alpha\leq k$$
and therefore is each composing vector field $X_\alpha$ a  $G$-invariant vector field on $M$.

Let us denote by $\xi_M$ the fundamental vector field for the action $\Phi$, associated to an element $\xi$ of the Lie algebra $\g$.  Recall that, if $G$ is connected, a function $f$ on $M$ is invariant if and only if $\xi_M(f)=0$ for all $\xi\in \mathfrak{g}$. Likewise, a vector field $X$ on $M$ is invariant if and only if $[X,\xi_M]=0$ for all $\xi\in \mathfrak{g}$.
In terms of the Lie derivative we had introduced in Section~\ref{sec2}, we obtain that ${\mathbf X}$ is invariant if and only if ${\mathcal L}_{\xi_M}{\mathbf X} = {\mathbf 0}$, for all $\xi\in\g$.

 \begin{definition} \label{defred}
The reduced $k$-vector field of a $G$-invariant $k$-vector field ${\mathbf X} =(X_\alpha)$ on $M$ is the $k$-vector field $\breve{\mathbf X}$ on $M/G$ whose composing parts are given by the reduced vector fields $\breve{X}_\alpha$ of $X_\alpha$, given by
\[
T\pi_M \circ X_\alpha = \breve{X}_\alpha \circ \pi_M.
\]
\end{definition}

From relation (\ref{condInteg0}) and the above definition of $\breve{X}_\alpha$ we can easily conclude:
\begin{prop} \label{m}
If $\phi$ is an integral section of an invariant $k$-vector field $\mathbf{X}$ on $M$, then  $\breve\phi=\pi_M\circ\phi$ is an integral section of the reduced $k$-vector field $\breve{\mathbf X}$ on $M/G$.
\end{prop}

\subsection{Integrability and curvature}\label{sec51}

In this section we consider  a principal fibre bundle $\pi_M: M \to M/G$. We wish to examine how the integrability of an invariant $k$-vector field on $M$ relates to the integrability of its reduced $k$-vector field on $M/G$.

From now we will use local coordinates on $M$ defined as follows. Let $U\subset M/G$ be an open set over which $M$ is locally trivial, so that
$(\pi_M)^{-1}(U)\simeq
U\times G$. We will use coordinates $(x^i,x^a)$ on a suitable open subset $(\pi_M)^{-1}(U)$  (containing $U\times \, e$) such that $(x^i)$ are coordinates on $U$, and $(x^a)$  are coordinates on the fibre $G$. Then, the local expression of the projection $\pi_M: M \to M/G$ is:
\begin{equation}\label{TriviCoord}
\begin{array}{ccc}
(\pi_M)^{-1}(U)=U\times G & \to & U \\
\noalign{\medskip} (x^I)=(x^i,x^a) & \mapsto & (x^i).
\end{array}
\end{equation}

In these coordinates, the left action of $G$ onto $(\pi_M)^{-1}(U)=U\times G$ is given
by
\[
\Phi_g(x,h)  =  (x,gh).
\]

We can write any $k$-vector field ${\mathbf X}$ on $M$ as
\begin{equation}\label{locX}
 X_\alpha= X^i_\alpha \fpd{}{x^i} + {\tilde X}^a_\alpha \fpd{}{x^a}.
  \end{equation}
If ${\mathbf X}$ is $G$-invariant then the functions $X_\alpha^i$ are invariant functions on $M$. They can therefore be identified with functions on $M/G$. The reduced $k$-vector field $\breve{\mathbf X} =({\breve X}_\alpha)$ we had defined in Definition~\ref{defred} is given by
\[
{\breve X}_\alpha= X^i_\alpha \fpd{}{x^i}.
\]

From (\ref{locgam0}) we know that the connection associated to the reduced $k$-vector field $\breve{\mathbf X}$ is given by right splitting  \begin{equation}\label{locgam1}\gamma^{\breve{\mathbf X}}  ([m],T^\alpha \derpar{}{t^\alpha}\Big\vert_{t})= T^\alpha \left(
 \derpar{}{t^\alpha}\Big\vert_{(t,[m])} +X^i_\alpha(m)\derpar{}{x^i}\Big\vert_{(t,[m])}\right)
 \end{equation}
of the short exact sequence
\[
0 \to \r^k \times T(M /G)    \to T(\r^k\times (M/G) )\to
(M/G) \times T\r^k
       \to 0.
\]
We will denote its  curvature by  $K^{\breve{\mathbf X}}$.

 \begin{prop}\label{CaractInteg}

(1) If  ${\mathbf X}$ is integrable, then so is also the reduced $k$-vector field  $\breve{\mathbf X}$, i.e.\ $K^{\breve{\mathbf X}}=0$.

(2) If  $\breve{\mathbf X}$ is integrable then the vector fields $[X_\alpha,X_\beta]$ take values in the vertical distribution  of $\pi_M$, which can be identified with $M\times\g$.
\end{prop}
\begin{proof}
Both properties easily follow from the fact that $T\pi_M\circ [X_\alpha,X_\beta] = [{\breve X}_\alpha, {\breve X}_\beta]\circ\pi_M$.
\end{proof}

When ${\breve{\mathbf X}}$ is integrable  the remaining vertical part of the bracket $[X_\alpha,X_\beta]$   is  locally given by \begin{equation}\label{rest}
\left[X^i_\alpha \fpd{}{x^i}, {\tilde X}^a_\beta \fpd{}{x^a}\right] - \left[X^i_\beta \fpd{}{x^i}, {\tilde X}^a_\alpha \fpd{}{x^a}\right] +     \left[ {\tilde X}^a_\alpha \fpd{}{x^a}, {\tilde X}^b_\beta \fpd{}{x^b}\right].
\end{equation}
In the calculation of these brackets one should take into account that all partial derivatives of the functions $X^i_\alpha$ with respect to variables $x^a$ vanish, because the components $X^i_\alpha$ are $G$-invariant.

In the remainder of this section, we will show that we can also give an interpretation of that vertical part (\ref{rest}), as the curvature of some connection.

Let $\breve\phi: \r^k \to M/G$ be an integral section of the reduced $k$-vector field $\breve{\mathbf X}$ on $M/G$. Consider the pull-back bundle $\pi_2: {\breve\phi}^*M \to \r^k$:
\[\xymatrix@=10mm{ {\breve\phi}^*M\ar[r]^-{} \ar[d]_-{\pi_2} & M \ar[d]^-{\pi_M}& \mbox{with}\, {\breve\phi}^*M=\left\{ (t,m)\, : \, \pi_M(m)= \breve\phi(t)\right\}\\
\r^k \ar[r]^-{ \breve\phi} & M/G &}
\]
This is a $G$-principal bundle. Let us use $i$ for the inclusion $i: \breve\phi^*M \to \r^k \times M$. We will use $p$ for a point in ${\breve\phi}^*M$, and $(t,m)$ for its inclusion $i(p)$ in $\r^k\times M$. Then, $\breve\phi(t)=\pi_M(m)$ in $M/G$.

\begin{lem}\label{RelIntphihat}\
\begin{enumerate}
\item If $\breve\phi: \r^k \to M/G$ is  an integral section of the reduced $\breve{\mathbf X}$
then  $\breve\phi$ satisfies
\begin{equation} \label{propcon}
\gamma^{\breve{\mathbf X}} (\breve\phi(t),v_t) = T_t\hat\phi (v_t), \quad \mbox{for all $v_t\in T_t\r^k$}.
\end{equation}
where $\hat\phi: \r^k \to \r^k\times M/G, t \mapsto (t,\breve\phi(t))$ and where $\gamma^{\breve{\mathbf X}}$ is  the  connection associated to $\breve{\mathbf X}$.

\item
The following diagram is commutative
\[\xymatrix@=10mm{ {\breve\phi}^*M\ar[r]^-{i} \ar[d]_-{\pi_2} & \rk\times M \ar[d]^-{\bar\pi_M=(Id_{\rk},\pi_M)}
&&&  p\ar[r]^-{i} \ar[d]_-{\pi_2} & (t,m) \ar[d]^-{\bar\pi_M=(Id_{\rk},\pi_M)}\\
\r^k \ar[r]^-{ \hat\phi} &\rk\times  M/G &&& t\ar[r]^-{ \hat\phi} &(t,[m])}
\] that is
\begin{equation}\label{pi2phihat}
\bar\pi_M\circ i = \hat\phi\circ\pi_2\,
\end{equation}
\end{enumerate}\end{lem}
\begin{proof} Both properties are immediate consequences of expression (\ref{locgam1}) and of the fact that  $\breve\phi: \r^k \to M/G$ is an integral section of  $\breve{\mathbf X}$.
\end{proof}

If ${\breve\phi}$ is locally $(t^\alpha) \mapsto (x^i = \phi^i(t))$, then locally
 $$i: (t^\alpha,x^a) \mapsto (t^\alpha,x^i=\phi^i(t),x^a)$$  that is to say, the pullback bundle structure naturally induces coordinates  $(t^\alpha,x^a)$ on ${\breve\phi}^*M$.
 In these coordinates, tangent vectors $V_p$ to ${\breve\phi}^*M$ (in a point $p$) are locally of the form
\begin{equation}\label{vp}V_{p} = T^\alpha \fpd{}{t^\alpha}\Big|_p + {\tilde Y}^a \fpd{}{x^a}\Big|_p \, .
\end{equation}
From the relations  \begin{eqnarray*}
T_pi\left(\fpd{}{t^\alpha}\Big|_p\right) &=&  \fpd{}{t^\alpha}\Big|_{(t,m)} +  \fpd{\phi^i}{t^\alpha}(t) \fpd{}{x^i}\Big|_{(t,m)}   =  \fpd{}{t^\alpha}\Big|_{(t,m)} + (X^i_\alpha\circ\breve\phi)(t) \fpd{}{x^i}\Big|_{(t,m)},\\ T_pi\left(\fpd{}{x^a}\Big|_p \right) &=& \ \fpd{}{x^a}\Big|_{(t,m)}
\end{eqnarray*}
we can deduce that
\begin{equation} \label{Vp}
T_pi(V_p)= T^\alpha \fpd{}{t^\alpha}\Big|_{(t,m)} + (X^i_\alpha\circ\breve\phi)(t)   T^\alpha \fpd{}{x^i}\Big|_{(t,m)} + {\tilde Y}^a \fpd{}{x^a}\Big|_{(t,m)} \in T_t\rk\times T_mM. \\
\end{equation}
Here we consider $X^i_\alpha$ as a function on $M/G$, and therefore $X^i_\alpha\circ\breve\phi$ as a function on $\r^k$.

Vertical vectors for the bundle $\pi_2$ at the point $p$ are those with $T^\alpha = 0$,  and may therefore be identified  with elements in $\r^k\times V_mM$,  where $ VM$ is the vertical distribution of $\pi_M:M\to M/G$. A connection on ${\breve\phi}^*M$ is therefore a splitting of the sequence
\[
0 \to V({\breve\phi}^*M) \equiv \r^k\times VM
 \to T({\breve\phi}^*M) \to {\breve\phi}^*M \times_{\r^k} T\r^k \to 0
\]
of vector bundles over ${\breve\phi}^*M$.

The connection map of the connection $\gamma^{\mathbf X}$ is a map  $\omega^{\mathbf X}: T(\r^k\times M) \to  \r^k \times TM$. For $V_{p} \in T(\breve\phi^*M)$, with  $i(p)=(t,m)$,
  it follows from  (\ref{locomega0}) and (\ref{Vp}) that
 \begin{equation}\label{cede1}\omega^{\bf X}(T_pi(V_p))=
\left(t , (  {\tilde Y}^a-{\tilde X}^a_\alpha(m) T^\alpha)\fpd{}{x^a}\Big\vert_{m}    \right), \quad \breve\phi(t)=\pi_M(m).
 \end{equation}

The second element is clearly $\pi_M$-vertical in $M$, and we may use it to define a connection on
${\breve\phi}^*M$.

\begin{definition}\label{defcon}
The principal connection $\gamma^{\breve\phi, {\mathbf X}}$ on ${\breve\phi}^*M$, defined as a connection map, is given by
 \begin{equation}\label{conmapphi}
\omega^{\breve\phi, {\mathbf X}}  (V_{p}) =
\left(t,\omega^{\mathbf X} (T_pi(V_{p}))\right) \in \r^k\times V_mM.
 \end{equation}
\end{definition}

In coordinates, if we represent $V_p$ as in (\ref{vp}),
then
\begin{equation}\label{conmapphi1}
\omega^{\breve\phi, {\mathbf X}}  (V_{p}) = \left(t,({\tilde Y}^a - {\tilde X}^a_\alpha(m)  T^\alpha) \fpd{}{x^a}\Big|_{m} \right).
 \end{equation}
Likewise, for the corresponding horizontal lift ${T}^h$ of a vector field $T=T^\alpha \partial / \partial{t^\alpha}$ on $\r^k$, we obtain, from (\ref{conmapphi1}),
\[
T^{h} = T^\alpha \left( \fpd{}{t^\alpha} + ({\tilde X}^a_\alpha\circ pr_2\circ   i)\fpd{}{x^a} \right) \in \vectorfields{{\breve\phi}^*M}.
\]
From now we shall denote the map $ pr_2\circ i$ by $\pi_1$.

The curvature of this connection is then (with  $T=T^\alpha   \partial / \partial {t^\alpha} $ and $S= S^\alpha \partial / \partial {t^\alpha} $ two vector fields on $\rk$):
\begin{eqnarray*}
K^{\breve\phi, {\mathbf X}}({T}, {S})&=&  -T^\alpha S^\beta \left(   \left(\fpd{({\tilde X}^a_\beta\circ \pi_1)}{t^\alpha} - \fpd{({\tilde X}^a_\alpha\circ \pi_1)}{t^\beta} \right)\fpd{}{x^a} \right.\\
      &&+ \left. \left[ ({\tilde X}^a_\alpha\circ \pi_1 ) \fpd{}{x^a}, ({\tilde X}^b_\beta\circ \pi_1 ) \fpd{}{x^b}\right] \right)
\end{eqnarray*}

If we take into account that
\[
\fpd{({\tilde X}^a_\beta\circ \pi_1)}{t^\alpha} = \left(\fpd{{\tilde X}^a_\beta}{x^m}\circ \pi_1\right) \fpd{\phi^m}{t^\beta} =  \left(\fpd{{\tilde X}^a_\beta}{x^m}\circ \pi_1\right) (X^m_\beta\circ\breve\phi) = \left(\fpd{{\tilde X}^a_\beta}{x^m} X^m_\beta\right)\circ \pi_1
\]
and that, since $X^i_\alpha$ are invariant functions, $\partial{X^i_\alpha}/\partial{x^a} =0$, we easily see that
\[T\pi_1(K^{\breve\phi, {\mathbf X}}({T}, {S})) = -T^\alpha S^\beta\left( \left[X^i_\alpha \fpd{}{x^i}, {\tilde X}^a_\beta \fpd{}{x^a}\right] - \left[X^i_\beta \fpd{}{x^i}, {\tilde X}^a_\alpha \fpd{}{x^a}\right] +     \left[ {\tilde X}^a_\alpha  \fpd{}{x^a}, {\tilde X}^b_\beta \fpd{}{x^b}\right] \right) \circ \pi_1.
\]

When we compare this to expression (\ref{rest}), we may conclude from Proposition~\ref{CaractInteg} that:
\begin{prop} \label{propint}
A $G$-invariant $k$-vector field  ${\mathbf X}$ is integrable if and only if
\begin{enumerate}
\item its reduced vector field  $\breve{\mathbf X}$ is integrable (i.e.\ its curvature as a connection vanishes) and
\item the curvature of the connection $\omega^{\breve\phi, {\mathbf X}}$ vanishes for each integral section $\breve\phi: \r^k \to M/G$ of $\breve{\mathbf X}$.
\end{enumerate}

If $\breve\phi$ is an integral section of $\breve{\mathbf X}$ then there will exist an integral section $\phi$ of ${\mathbf X}$ that projects on $\breve\phi$ provided that  the curvature of the connection $\omega^{\breve\phi, {\mathbf X}}$ vanishes.
\end{prop}

 In Section~\ref{sec6}  we will give a method that will enable us to actually reconstruct such an integral section $\phi$.

 We can, in particular, use the proposition above to characterize the integrability of the Euler-Lagrange equations (\ref{lfield}) of an invariant Lagrangian $L$. For that case, we have to take  $M=T^1_kQ$ and ${\mathbf X} = {\mathbf \Gamma}$, a Lagrangian $k$-vector field.

\subsection{Decomposition by making use of a principal connection} \label{se52n}

We will show in this section that, by making use of a principal connection $\omega^M: TM \to TM$ on $\pi_M:M\to M/G$, we can re-express the vanishing of the curvature of the connection $\omega^{\breve\phi, {\mathbf X}}$ in more convenient terms. We also show that, in stead of defining this connection directly, as we did in  Definition~\ref{defcon}, we could also have constructed it in two consecutive steps.

Suppose we are given a principal connection on the principal bundle $\pi_M: M \to M/G$. We will consider  three sets of vector fields $\{X_i\}$, $\{\widetilde{E}_a\}$ and $\{\widehat{E}_a\}$ on $M$. The first set, $\{X_i\}$, is given by the horizontal lifts of a coordinate basis of
vector fields ${\partial}/{\partial x^i}$ on $M/G$ by the given principal connection. These
vector fields are $G$-invariant by construction, and they form a basis of the
horizontal subspace at any point. The two other sets of vector fields, $\{\widetilde{E}_a\}$ and $\{\widehat{E}_a\}$, will both form a basis for the vertical space of $\pi_M$ at each point.

The vector fields $\{\widetilde{E}_a\}$  are the fundamental vector fields on
$M$, associated to a basis $\{E_a\}$ of the Lie algebra
$\mathfrak{g}$. They are in general not invariant vector fields. Since they are vertical by construction, we can write
\begin{equation}\label{K}
\widetilde{E}_a=K^b_a \derpar{}{x^b},
\end{equation}
for some non-singular matrix-valued function $(K^b_a)$. The vector fields  $\widehat{E}_a$ in the last set are defined as
\begin{equation}\label{Ehat}
\widehat{E}_a(x,g)=({ad_{g^{-1}}E_a})_M(x,g),
\end{equation}
where we are using the local trivialization (\ref{TriviCoord}) and where the notation $\xi_M$ refers again to the fundamental vector field of $\xi\in\g$. One easily verifies that these vector fields are all invariant.
The relation between $\widehat{E}_a$ and $ \widetilde{E}_a$ can be
expressed as
\begin{equation} \label{A}
\widehat{E}_a(x,g)\, = \, A^b_a(g)\widetilde{E}_b(x,g),
\end{equation}
where $(A^b_a(g))$ is the matrix
representing $ad_{g^{-1}}: \mathfrak{g}\to  \mathfrak{g}$ with respect to the basis $\{E_a \}$ of
$\mathfrak{g}$. In particular $A^b_a(e)=\delta^b_a$.

If we set
\begin{equation} \label{gamma}
X_i=\ds\frac{\partial }{\partial  x^i}-\gamma^b_i(x^i,x^a)\widehat{E}_b
\end{equation}
the invariance of $X_i$ amounts to $\partial \gamma^b_i/\partial x^a=0$. One easily verifies that the Lie brackets of the vector fields of interest (see e.g.\ \cite{MC}) are as follows:
\begin{equation}\label{e2}
\begin{array}{lcl}
\, [\widetilde{E}_a,\widetilde{E}_b ]=-C_{ab}^c\widetilde{E}_c, & [\widehat{E}_a,\widehat{E}_b ]= C_{ab}^c\widehat{E}_c, &
[X_i,\widetilde{E}_a ]=0, \\[1mm] \,   [X_i, \widehat{E}_a ]= \Upsilon^b_{ia} \widehat{E}_b,  &
[X_i ,X_j]=-K^a_{ij} \widehat{E}_a,   &  [\widetilde{E}_a ,\widehat{E}_b]=0.
\end{array}
\end{equation}
Here  $C_{ab}^c$ are the structure constants of the Lie algebra $\mathfrak{g}$, $K^a_{ij}$ are the components of the curvature of the principal connection (with respect to the vertical frame ${\widehat E}_a$) and $\Upsilon^b_{ia} = -\gamma_i^c C^b_{ca}$.

The {\em vertical lift} allows us to identify invariant vertical vector fields on $M$ with sections of the adjoint bundle $\bar\g = (M\times \g)/G \to M/G$, as follows: Let ${\bar E}_a$ be the local section $x \mapsto [(x,e), E_a]_G$, then $(X^a {\bar E}_a)^v = X^a {\widehat E}_a$ (see e.g.\ \cite{RRA}). The horizontal lift of the principal connection maps the vector field $\breve X = X^i \partial/\partial {x^i}$ on $M/G$ to the invariant horizontal vector field $X^h = X^iX_i$ on $M$. The decomposition of a  $G$-invariant vector field $X$  into a horizontal and a vertical part is then
\begin{equation} \label{dec}
X = (\breve X)^h + (\bar X)^v = X^i X_i + X^a {\widehat E}_a
\end{equation}
for a certain $\breve X \in \vectorfields{M/G}$ and $\bar X \in Sec(\bar\g \to M/G)$. Both coefficients $X^i$ and $X^a$ can be identified with $G$-invariant functions on $M$, and therefore with functions on $M/G$.

For two $G$-invariant vector fields $X$ and $Y$ on $M$, the bracket $[X,Y]$ is again $G$-invariant and one can verify that
\begin{equation}\label{dec2}
[X,Y] = ([\breve X,\breve Y])^h + ( \nabla_{\breve X}\bar Y -  \nabla_{\breve Y}\bar X + [\bar X,\bar Y] - {\bar K}^M(\breve X, \breve Y)    )^v,
\end{equation}
see e.g.\ Theorem 5.2.4 in \cite{Cendra}, or \cite{Mestdag}. Here
\begin{equation}\label{CurvHor}
({\bar K}^M(\breve X, \breve Y))^v= -\omega^M([{\breve X}^h, {\breve Y}^h])
 \end{equation}
 is the curvature of the connection $\omega^M$, if one takes the identification between sections of the adjoint bundle and vertical vector fields into account. The bracket $[\bar X,\bar Y]$ is the Lie bracket on sections of the adjoint bundle $\bar\g \to M/G$ (which is a Lie algebra bundle), given by
\[
[\bar X,\bar Y]^v = [{\bar X}^v,{\bar Y}^v]  \qquad \mbox{or} \qquad [{\bar E}_a,{\bar E}_b] = C_{ab}^c {\bar E}_c,
\]
 and the connection $\nabla$ is the induced connection on the adjoint bundle, given by
\[
(\nabla_{\breve X}\bar Y)^v = [(\breve X)^h,{\bar Y}^v]\qquad \mbox{or} \qquad \nabla_{\fpd{}{x^i}}{\bar E}_a  = \Upsilon_{ia}^b {\bar E}_b.
\]

If ${\mathbf X} = (X_\alpha)$ is $G$-invariant $k$-vector field on $M$, then the decomposition (\ref{dec}) defines a reduced $k$-vector field $\breve{\mathbf X} = ({\breve X}_\alpha)$ on $M/G$ and a section $\bar{\mathbf X} = ({\bar X}_\alpha)$ of ${\bar\g}^k \to M/G$.

\begin{prop}\label{newprop} Given a principal connection $\omega^M$, a $G$-invariant $k$-vector field ${\mathbf X}$ is integrable if and only if
\begin{enumerate}
\item $\breve{\mathbf X}$ is integrable and
\item $
\nabla_{{\breve X}_\alpha}{\bar X}_\beta -  \nabla_{{\breve X}_\beta}{\bar X}_\alpha + [{\bar X}_\alpha,{\bar X}_\beta] - {\bar K}^M({\breve X}_\alpha, {\breve X}_\beta) =0.
$
\end{enumerate}
\end{prop}
\begin{proof}
 The integrability of  ${\mathbf X}$ is measured by the vanishing of the brackets $[X_\alpha,X_\beta]$. But we have now,
\[
[X_\alpha,X_\beta] = ([{\breve X}_\alpha,{\breve X}_\beta])^h + ( \nabla_{{\breve X}_\alpha}{\bar X}_\beta -  \nabla_{{\breve X}_\beta}{\bar X}_\alpha + [{\bar X}_\alpha,{\bar X}_\beta] - {\bar K}^M({\breve X}_\alpha, {\breve X}_\beta)    )^v.
\] 
\end{proof}

The first condition means that the curvature of $\breve{\mathbf X}$, regarded as a connection, must vanish.
 If we set
\begin{equation} \label{dec3}
X_\alpha =  ({\breve X}_\alpha)^h + ({\bar X}_\alpha)^v = X^i_\alpha X_i + X^a_\alpha {\widehat E}_a,
\end{equation}
then from (\ref{K}), (\ref{A}) and (\ref{gamma}), the relation with the notations in the preceding paragraph is
\begin{equation} \label{tildeX}
{\tilde X}^a_\alpha = (X^c_\alpha - X^i_\alpha\gamma_i^c) A^b_c K^a_b.
\end{equation}

In coordinates, the second condition in Proposition~\ref{newprop} is
\begin{equation} \label{seccond}
{\breve X}_\alpha(X^b_\beta)  - {\breve X}_\beta(X^b_\alpha) + (X^i_{\alpha} X^a_\beta   - X^i_{\beta} X^a_\alpha)   \Upsilon_{ia}^{b}   + C^b_{ac}X^a_\alpha X^c_\beta  - K^b_{ij}X^i_\alpha X^j_\beta  =0.
\end{equation}
From Proposition~\ref{propint} we know that we may identify this expression with the vanishing of the curvature of the connection $\gamma^{\breve\phi, {\mathbf X}}$ for each integral section $\breve\phi$, i.e.\ it is equivalent with the condition (\ref{rest}).

We will need the expression (\ref{seccond}) later in Section~\ref{sec53}. For later comparison with \cite{Ellis}, we show how one can use the principal connection $\omega^M$ to split the connection $\gamma^{\breve\phi, {\mathbf X}}$ in two parts.

Consider again the pull-back bundle ${\breve\phi}^*M \to \r^k$.
Let us denote the map $p\in  {\breve\phi}^*M \to m\in M$, as before,  by $\pi_1$.
 Then we can define a new principal connection $\omega^{\breve\phi}$ on ${\breve\phi}^*M$ as
\begin{equation} \label{con2}
\omega^{\breve\phi} (V_p) = \omega^M (T_p\pi_1(V_{p})), \quad \forall V_{p} \in T({\breve\phi}^*M).
\end{equation}
Tangent vectors to ${\breve\phi}^*M $ (in a point $i(p)=(t,m)$ with $\phi(t) =\pi(m) \in M/G$) can now be represented in the form
\[
Ti(V_p) = T^\alpha \fpd{}{t^\alpha}\Big|_t + (X^i_\alpha\circ\breve\phi)  T^\alpha X_i(m) + Y^a {\widehat E}_a(m).
\]
The relationship between $Y^c$ and ${\tilde Y}^a$ of expression (\ref{Vp}) is then
 \begin{equation} \label{con30}
  {\tilde Y}^a  = (Y^c- T^\alpha X^i_\alpha\gamma_i^c) A^b_c K^a_b\, .   \end{equation}

A local expression for this connection is then
\[
\omega^{\breve\phi} (V_p) = Y^a {\widehat E}_a(m).
\]

Suppose that we are now also given a section $\bar{\mathbf X} = ({\bar X}_\alpha = X^a_\alpha {\bar E}_a)$ of $\bar \g^k \to M/G$. We can vertically lift it to the section  $(X^a_\alpha\circ\breve\phi) {\bar E}_a$ of ${\breve\phi}^*M$ and add it to the connection $\omega^{\breve\phi}$ to form a new connection on ${\breve\phi}^*M$, with
\begin{equation} \label{con3}
\omega^{\breve\phi, \bar{\mathbf X}}(V_p) = (Y^a - (X^a_\alpha \circ\breve\phi) T^\alpha ) {\widehat E}_a(m).
\end{equation}
From (\ref{tildeX}),   (\ref{con30}) and   (\ref{con3}), we see that the connection $\omega^{\breve\phi, \bar{\mathbf X}}$ is the same as the connection $\omega^{\breve\phi, {\mathbf X}}$ we had introduced in the paragraphs above.

\section{Lagrangian $k$-symplectic field theory} \label{sec3}

In this section, we recall the Lagrangian $k$-simplectic formalism. For a regular Lagrangian, the solutions of the field equations are given by the integral sections of some $k$-vector fields, the so-called Lagrangian {\sc sopdes}. They represent a generalization of the well-known concept of a {\sc sode} vector field.

 \subsection{Canonical operations on  $T^1_kQ$} \label{sec31}

  In this section we will assume that $Q$ is an $n$-dimensional differentiable manifold, whose local coordinates are given by $(q^A)$. We will denote the natural coordinates of $T^1_kQ$ by $(q^A,u^A_\alpha)$. In the next paragraphs we briefly recall some canonical objects and structures that can be defined on $M=T^1_kQ$. Most of them find their natural analogue on a tangent bundle,  when $k=1$ (see e.g. \cite{CP,LR} for that case).

We will assume throughout that a point ${\mathbf v} = (q;{v_1},\ldots,{v_k})\in T_k^1Q$ is given, with $\tau^1_Q({\mathbf v}) = q \in Q$. For a tangent vector $Z_q\in T_qQ$, we define its
{\sl vertical $\alpha$-lift  at ${\mathbf v}$}, $Z^{V_\alpha}_{{\mathbf v}}$, as the vector tangent to the
fiber $(\tau^1_Q)^{-1}(q)\subset T_k^1Q$,  given by
  $$
Z^{V_\alpha}_{{\mathbf v}}=
\frac{d}{ds}({v_1},\ldots,{v_{\alpha-1}},v_{\alpha}+sZ,{v_{\alpha+1}},\ldots,{v_k})\Big\vert_{s=0} \in T_{{\mathbf v}}(T^1_kQ). $$

We can, of course, extend this operation to the level of vector fields. If $Z=Z^A \partial /\partial q^A$ is a vector field on $Q$, then its $\alpha$-th vertical lift $Z^{V_\alpha}$ is the vector field on $T^1_kQ$ whose local expression is
\begin{equation} \label{valift}
Z^{V_\alpha}=Z^A\frac{\partial}{\partial u^A_\alpha}\, .
\end{equation}

There is a corresponding notion of a complete lift $Z^C$. If the vector field $Z$ on $Q$ has local $1$-parametric group
of transformations $\varphi_t \colon Q \to Q$, then the local
$1$-parametric group of transformations $T^1_k\varphi_t\colon T^1_kQ
\to T^1_kQ$ generates a vector field $Z^C$ on $T^1_kQ$, the {\sl complete lift} of $Z$ to $ T^1_kQ$. Its local expression is
\beq \label{clift}
Z^C=Z^A\frac{\partial}{\partial q^A}+u^A_\alpha \frac{\partial Z^B}
{\partial q^A}\frac{\partial}{\partial u^B_\alpha} \,.
\eeq
One may easily establish the following properties for the brackets of complete and vertical lifts:
\begin{equation}\label{cv}
[X^C,Y^C]= [X,Y]^C, \quad [X^C,Y^{V_\alpha}]= [X,Y]^{V_\alpha}, \quad  [X^{V_\alpha},X^{V_\beta}]=0.
\end{equation}

 With a {\sl local frame on $Q$} we will mean a basis for the $\cinfty{Q}$-module structure of the set of vector fields on $Q$, that is to say: If  $\{Z_A\}$ is a frame on $Q$, then each vector field $Z$ on $Q$ can be written as $Z=Z^A Z_A$, for some functions $Z^A$ on $Q$. Likewise, each tangent vector $v_q \in T_qQ$ can be decomposed as $v_q = v^A Z_A(q)$, for some real numbers $v^A$.
From the local expressions (\ref{valift}) and (\ref{clift}), we can easily conclude that:
\begin{prop}\label{frame}
If $\{Z_A\}$ is any local frame on $Q$, then $\{Z_A^C, Z_A^{V_\alpha}\}$ is a local frame on $T^1_kQ$.
\end{prop}

The {\sl canonical $k$-tangent structure} on $T^1_kQ$ is the set
 of (1,1) tensor fields $(S^1,\ldots,S^k)$ defined by
 $$
S^\alpha({\mathbf v})(Z_{{\mathbf v}})= (T_{\mathbf v}\tau^1_Q(Z_{{\mathbf v}}))^{V_\alpha}_{\mathbf v}, \quad
   {\mathbf v}\in T^1_kQ,\quad Z_{{\mathbf v}}\in T_{{\mathbf v}}(T^1_kQ).
   $$
Alternatively we can define $S^\alpha$, for its action on vector fields, as the unique (1,1)-tensor field on $T^1_kQ$ for which $S^\alpha(X^C) = X^{V_\alpha}$ and $S^\alpha(X^{V_\beta}) = 0$. Its local expression is
\begin{equation}\label{salpha} S^\alpha=\frac{\partial}{\partial u^A_\alpha} \otimes \d
q^A.
\end{equation}

The {\sl Liouville vector field} $\Delta\in \vectorfields{T^1_kQ}$ is the infinitesimal generator of
the  flow
 $$
\psi\colon\r\times T^1_kQ\longrightarrow T^1_kQ, \quad
  \psi(s,v_{1_{q}}, \ldots ,v_{k_{q}})= (e^s   v_{1_{q}},\ldots , e^s   v_{k_{q}}).
$$
In local coordinates it takes the form $
 \Delta = u^A_\alpha \partial/ \partial {u^A_\alpha}\ .
$

\subsection{Second-order partial differential equations}

Maps like $\phi\colon U_0\subset \r^k \rightarrow Q$ will play the role of the fields of the theory. The differential equations of interest, however, are second-order partial differential equations, and will be defined on $M=T^1_kQ$, rather than on $Q$. We will turn next to the characterization of those integrable
$k$-vector fields on $T^1_kQ$ which have the property that all their integral sections
are first prolongations $\phi^{(1)}$ of maps
$\phi\colon\r^k\to Q$.

\begin{definition}
\label{sode0} A  second-order partial differential equation field
({\sc sopde} from now on) is a $k$-vector field
$\mathbf{\Gamma}$ on $M=T^1_kQ$ which is
a section of the projection $T^1_k\tau^1_Q\colon
T^1_k(T^1_kQ)\rightarrow T^1_kQ$; that is,  it satisfies
$$
T^1_k\tau^1_Q\circ\mathbf{\Gamma}={\rm Id}_{T^1_kQ},
$$
for $\tau^1_Q: T^1_kQ\to Q$.
\end{definition}
For ${\mathbf \Gamma} =(\Gamma_\alpha)$ this definition is equivalent with the property that, for all  ${\mathbf w} =(w_\alpha) \in T^1_kQ$,
$$
T_{\mathbf w}\tau^1_Q(\Gamma_\alpha({\mathbf w}))= w_{\alpha}.
$$
For $k=1$, the definition of a {\sc sopde} reduces to that of a  second-order ordinary
differential equation field (often called {\sc sode}).

In local coordinates we obtain that the local expression of a {\sc
sopde} $\mathbf{\Gamma}$ is
\begin{equation}
\label{localsode1} \Gamma_\alpha= u^A_\alpha\frac{\partial}
{\partial q^A}+ (f_\alpha)^B_\beta \frac{\partial} {\partial
u^B_\beta}, 
\end{equation}
for some functions $(f_\alpha)^B_\beta \in \cinfty{T^1_kQ}$.

If $\psi\colon\r^k \to T^1_kQ$,  locally given by
$\psi(t)=(\phi^A(t),\psi_\alpha^A(t))$, is an integral section of a {\sc sopde}
$\mathbf{\Gamma}$ then we obtain from Definition~\ref{integsect} and expression (\ref{localsode1})  that
\[
 \frac{\partial\phi^A} {\partial t^\alpha}\Big\vert_t=\psi^A_\alpha(t)\quad ,\quad
\frac{\partial\psi_\alpha^B} {\partial
t^\beta}\Big\vert_t=(f_\alpha)^B_\beta(\psi(t))\, .
\]
 From this, we obtain the following proposition, see \cite{BBS, rsv07}.
\begin{prop}
 \label{sope1}
Let $\mathbf{\Gamma}$ be an integrable
{\sc sopde}. Each integral section $\psi$ of ${\mathbf \Gamma}$
is the first prolongation $\phi^{(1)}$ of its projection
$\phi=\tau^1_Q\circ\psi\colon\r^k\to Q$ onto $Q$. Moreover,
 $\phi$ is a solution of the system of second order partial
differential equations given by
\begin{equation}
\label{nn1}
 \frac{\partial^2 \phi^A}{\partial t^\alpha\partial t^\beta}(t)=
(f_\alpha)^A_\beta\left(\phi^B(t),\frac{\partial\phi^B}{\partial
t^\gamma}(t)\right). 
\end{equation}
Conversely, if $\phi\colon\r^k \to Q$ is a  map satisfying
(\ref{nn1}), then its prolongation $\phi^{(1)}$ is an integral section of
$\mathbf{\Gamma}$.
\end{prop}

 From (\ref{nn1}) we deduce that if $\mathbf{\Gamma}$ is an
integrable {\sc sopde} then $(f_\alpha)^A_\beta=(f_\beta)^A_\alpha$ for
all choices $\alpha,\beta,=1,\ldots, k$. We will call $\phi$ a solution of $\mathbf{\Gamma}$, whenever $\phi^{(1)}$ is one of its integral sections.

\subsection{Lagrangian  {\sc sopde}s} \label{sec32}

We are now all set to describe the Lagrangian field theory of interest: the $k$-symplectic formalism.  In this context, a Lagrangian is a function $L$ on $T^1_kQ$. By using the $k$-tangent structure $S^\alpha$, we  may  introduce a family of $k$ one-forms  $\theta_L^\alpha=  \d L \circ S^\alpha$ 
and $k$ two-forms $\omega_L^\alpha=-\d\theta_L^\alpha$ on $T^1_kQ$
with local expressions
\begin{equation}
\label{locthetas}
  \theta_L^\alpha=  \frac{\partial L}{\partial u^A_\alpha }\, dq^A\, , \quad \omega_L^\alpha=dq^A \wedge
  d\left(  \frac{\partial L}{\partial u^A_\alpha }\ \right).
  \end{equation}
We may also introduce the {\sl energy  function}
$E_L=\Delta(L)-L\in C^\infty(T^1_kM)$.

\begin{definition}
 The Lagrangian $L\colon T^1_kQ\to \r$ is
said to be regular if the matrix
 $\displaystyle \left(\frac{\partial^2 L}{\partial u^A_\alpha \partial u^B_\beta}\right)$ is non-singular
at every point of $T^1_kQ$.
\end{definition}

For the rest of the paper, we will assume that $L$ is regular. In \cite{BBS,fam,rsv07}   it has been shown that, under that condition, all $k$-vector fields  ${\mathbf \Gamma} =(\Gamma_\alpha)$  on $T^1_kQ$ that satisfy the condition
\[   \imath_{ \Gamma_\alpha}\omega_L^\alpha=\d E_L,
\] must be {\sc sopde}s.  Moreover, if ${\mathbf \Gamma}$ is a {\sc sopde}, the above relation  is equivalent with
\[
{\mathcal L}_{\Gamma_\alpha} \theta^\alpha_L - \d L = 0\, ,
\]
see Proposition~2.11 in \cite{BBS}.
\begin{definition} A {\sc sopde} $\Gamma$ will be called a {\sl Lagrangian {\sc sopde}} for $L$ if it satisfies the above equation.
\end{definition}
Given that $[\Gamma_\alpha, Z^{C}]= W_\alpha$ and $[\Gamma_\alpha, Z^{V_\beta}] = -\delta_\alpha^\beta Z^C + V_{\alpha}^\beta$, where all $V_{\alpha}^\beta$ and $W_\alpha$ are vertical vector fields for the projection $\tau^1_Q: T^1_kQ \to Q$, the above relation, when applied to a complete lift $Z^C$ satisfies
\begin{eqnarray*}
0 &=& ({\mathcal L}_{\Gamma_\alpha} \theta^\alpha_L - \d L) (Z^C) = \Gamma_\alpha(\theta_L^\alpha(Z^C)) - \theta^\alpha_L([\Gamma_\alpha,Z^C]) - Z^C(L) \\
&=& \Gamma_\alpha(Z^{V_\alpha}(L)) - Z^C(L).
\end{eqnarray*}

When applied to a vertical lift $Z^{V_\beta}$, we simply get an identity "0=0". In view of Proposition~\ref{frame}, we can conclude therefore:

\begin{prop}\label{frameprop}
A {\sc sopde} ${\mathbf \Gamma} =(\Gamma_\alpha)$ is Lagrangian for a regular Lagrangian $L$ if, and only if, for each vector field $Z$ on $Q$,
\[
\Gamma_\alpha(Z^{V_\alpha}(L)) - Z^C(L) = 0,
\]
or, equivalently, if for each local frame $\{Z_A\}$ of vector fields on $Q$,
\begin{equation}\label{elZ}
\Gamma_\alpha(Z_A^{V_\alpha}(L)) - Z_A^C(L) = 0, \qquad A=1\ldots n.
\end{equation}
\end{prop}

In particular, if we take the standard frame $\{ \partial /\partial q^A\}$ on $Q$, the equations (\ref{elZ}) become
\[
 \Gamma_\alpha \left(\ds\frac{ \partial L} {\partial u^A_\alpha}\right)-
\ds\frac{ \partial L} {\partial q^A}=0 .
\]
From this, it is easy to see that, if $\phi^{(1)} = (\phi^A, \partial\phi^A/\partial t^\alpha)$ is an integral section of ${\mathbf \Gamma}$, then it must satisfy
\[
 \displaystyle\frac{\partial^2 L}{\partial q^B \partial u^A_\alpha}  \frac{\displaystyle\partial\phi^B} {\displaystyle\partial
t^\alpha}  +
 \ds\frac{\partial^2 L}{\partial u_\beta^B \partial u^A_\alpha  }  \frac{\displaystyle\partial^2\phi^B} {\displaystyle\partial
t^\alpha \partial t^\beta}   =  \ds\frac{\partial  L}{\partial q^A
}.
  \]
which are the Euler-Lagrange equations (\ref{lfield}) of the field theory.

 In what follows, however, we will rather need the equivalent expressions (\ref{elZ}) of these equations, expressed in the so-called {\sl quasi-$k^1$-velocities} of a given frame $\{Z_A\}$ on $Q$.

\begin{definition} The quasi-$k^1$-velocities of the element ${\mathbf v}  = (v_\alpha)\in T^1_kQ$ with $\tau^1_Q({\mathbf v}) = q$ along the local frame $\{Z_A\}$ on $Q$ are the real numbers $v^A_\alpha$, for which each $v_\alpha$ can be written as
\[
v_\alpha = v^A_\alpha Z_A(q).
\]
\end{definition}
We can therefore use $(q^A,v^A_\alpha)$ as (non-natural) coordinates in $T^1_kQ$. If $Z_A = Z_A^B \partial /\partial q^B$, their relation to the natural induced coordinates $(q^A,u^A_\alpha)$ is
\begin{equation} \label{coordrelation}
 u^B_\alpha =v^A_\alpha Z^B_A\, .
\end{equation}

Assume that ${\mathbf X} =(X_\alpha)$,
with
\[
X_\alpha =X^A_\alpha Z_A^C+ (Y_\alpha)_\beta^A Z^{V_\beta}_A,
\]
is a $k$-vector field  on $T^1_kQ$.  One easily verifies that
a section $\phi(t)=(q^A=\phi^A(t),v^A_\beta=\phi^A_\beta(t))$, given in quasi-$k^1$-velocities, is an integral section of ${\mathbf X}$ if it satisfies
\begin{equation} \label{intsectionframe}
\fpd{\phi^A}{t^\alpha} = (X^B_\alpha Z^A_B)\circ\phi, \qquad \fpd{\phi^A_\beta}{t^\alpha} = ((Y_\alpha)^A_\beta - R^A_{BC} X^B_\alpha v^C_\beta)\circ\phi,
\end{equation}
where $[Z_B,Z_C] = R_{BC}^DZ_D$ is the 'curvature' of the frame.

\begin{lem} \label{sopdelem} A {\sc sopde} $\Gamma$, written in terms of quasi-$k^1$-velocities, takes the form
 \begin{equation}\label{gammaz}
\Gamma_\alpha=v^A_\alpha Z_A^C+ (\Gamma_\alpha)_\beta^A Z^{V_\beta}_A
\end{equation}
 for some functions $ ({\Gamma}_A)^\alpha_B$ on $T^1_kQ$.
\end{lem}
\begin{proof}
This is a consequence of Proposition~\ref{frame} and of the properties  $T\tau^1_Q \circ Z_A^C = Z_A \circ \tau^1_Q$ and  $T\tau^1_Q \circ Z_A^{V_\beta}= 0$.
\end{proof}

To end this section, we say a few words about regularity in terms of a non-standard frame.

\begin{prop} \label{rank}
Let $\{Z_A\}$ be a local frame of vector fields on $Q$. A Lagrangian $L$ is regular if and only if the $(nk)$-square matrix of functions $(Z^{V_{\alpha}}_A(Z^{V_\beta}_B(L)))$ on $T^1_kQ$ has maximal rank.
\end{prop}
\begin{proof}
If we set $Z_A = Z_A^C \partial /\partial {q^C}$, then the matrix $Z=(Z_A^B )$ of functions on $Q$ is non-singular in each point. We have
\[
\Big(Z^{V_{\alpha}}_A(Z^{V_\beta}_B(L))\Big)=\Big( Z_A^C \spd{L}{u^C_\alpha}{u^E_\beta}Z^E_B \Big) = \Big(Z_A^C \delta^\gamma_\alpha \Big) \Big( \spd{L}{u^C_\gamma}{u^E_\epsilon}\Big) \Big(Z^E_B \delta^\epsilon_\beta\Big),
\]
where the right-hand side can be interpreted as the matrix multiplication of 3 $(nk)$-square matrices. Given that the determinant of the matrix $(Z_A^C \delta^\gamma_\alpha \big)$ is $k\det(Z) \neq 0$, we easily see that also the determinant of the matrix in the left-hand side never vanishes. 
\end{proof}

\section{Symmetry reduction of a Lagrangian $k$-vector field}\label{sec4}

In this section we show that, if the Lagrangian is $G$-invariant, then so are its Lagrangian {\sc sopde}s. The integral sections of the reduced {\sc sopde} will provide the Lagrange-Poincar\'e equations. We finish this section with a study of the integrability of a $G$-invariant {\sc sopde} and  its reduced {\sc sopde}.

\subsection{Invariant Lagrangian {\sc sopde}s} \label{newsec51}

Suppose we are given an action $\Phi$ by $G$ on $Q$. As before, we will denote by $ \xi_Q$ the fundamental vector field corresponding to $\xi\in \mathfrak{g}$. We have seen in Section~\ref{sec41} that this action may be lifted to one $T^1_kQ$. From the definition of a complete lift in Section~\ref{sec31}, it follows that the fundamental vector fields $\xi_{T^1_kQ}$ of this action are actually the complete lifts $\xi_Q^C$  of the fundamental vector fields  $\xi_Q$ of the action on $Q$.

From Proposition~\ref{frameprop} we know that the Lagrangian {\sc sopde}s ${\mathbf \Gamma}$ are those that satisfy the Euler-Lagrange equations (\ref{elZ}). Without loss of generality we may suppose that the local frame  $\{Z_A\}$ consists of only invariant vector fields (for example, we can use the invariant frame $\{X_i,{\hat E}_a\}$ that we had introduced in Section~\ref{se52n}).

\begin{lem} \label{lem51}  If the frame $\{Z_A\}$ on $Q$  is invariant then the frame $\{Z_A^C,Z_A^{V_\alpha}\}$ on $T^1_kQ$ is also invariant, with respect to the lifted action on $T^1_kQ$.
\end{lem}
\begin{proof}
This follows from the bracket relations (\ref{cv}):
$$
[\xi^C_Q,Z_A^C]= [\xi_Q,Z_A]^C=0, \qquad   [\xi^C_Q,Z_A^{V_\beta}]= [\xi_Q,Z_A]^{V_\beta}=0.  $$
\end{proof}

We will use coordinates $(q^A)=(q^i,q^a)$ on $Q$ that are adapted to the principal fibre bundle structure $Q\to Q/G$, as explained in Section~\label{sec52n} (with now $M=Q$).
If the quasi-$k$-velocities on $T^1_kQ$ with respect to the frame $\{Z_A\}$ are given by $v^A_\alpha$, then the couple $(q^i,q^a,v^A_\alpha)$ represents coordinates on $T^1_kQ$.
We shall show that the coordinate functions $q^i,v^A_\alpha$ are $G$-invariant functions on $T^1_kQ$ (i.e.\ $\xi_Q^C(q^i)=0$ and $\xi_Q^C(v^A_\alpha)=0$).

\begin{definition}
 Let $\theta$ be a  $1$-form on $Q$. We define   linear functions $\overrightarrow{\theta_\alpha}$  on $T^1_kQ$, such that, for ${\mathbf v} =(v_\alpha)\in T^1_kQ$,
$$
\overrightarrow{\theta_\alpha}({\mathbf v}) =
   \theta(v_{\alpha}).
   $$
\end{definition}If in local coordinates $\theta=\theta_A \, dq^A$,
then
\begin{equation}\label{thetaarrow}\overrightarrow{\theta_\alpha} =
\theta_A\, u^A_\alpha.
\end{equation}
From (\ref{valift}), (\ref{clift}) and (\ref{thetaarrow})  we can conclude the following relations.

\begin{lem}\label{pztfle}
Let $Z$ be a vector field on  $Q$, $f$ a function  on
$Q$, and $\theta$  a $1$-form on $Q$. Then
\[
Z^C(f)  =  Z(f), \quad Z^{V_\beta}(f)= 0, \quad
Z^C(\overrightarrow{\theta_\alpha}) = \overrightarrow{({\mathcal L}_Z\theta)_\alpha}, \quad
Z^{V_\beta}(\overrightarrow{\theta_\alpha})   =   \delta^\beta_\alpha \theta(Z).
\]
\end{lem}

If  $\{Z_A\}$ is a local frame on $Q$
and if $\{\theta^A\}$ is its dual basis, then the local quasi-$k$-velocities
 $v_A^\alpha$ of $T^1_kQ$ can in fact be represented by the linear functions $
v^A_\alpha = \overrightarrow{(\theta^A)_\alpha}$.

\begin{lem}  For   a local invariant frame on $Q$ the functions $q^i$ and $v^A_\alpha$ are $G$-invariant on $T^1_kQ$.
\end{lem}
\begin{proof}
The fundamental vector fields $\xi_Q$ are vertical with
respect to the projection $\pi_Q:Q \to Q/G$, and therefore $\xi_Q^C(q^i) = \xi_Q(q^i)=0$.
 From Lemma~\ref{pztfle} we obtain
$$\xi_Q^C (v^A_\alpha) =  \xi_Q^C( \overrightarrow{\theta^A_\alpha})
= \overrightarrow{({\mathcal L}_{
\xi_Q}\theta^A)_\alpha}.
$$
Since
$$
({\mathcal L}_{ \xi_Q}\theta^A)(Z_B) =  {\mathcal L}_{\xi_Q}(\theta^A(Z_B)) - \theta^A ({\mathcal L}_{ \xi_Q} Z_B) = 0,
$$
we obtain that $\xi_Q^C (v^A_\alpha)=0$. 
\end{proof}

For the remainder of the paper we will suppose that the Lagrangian $L$ is invariant under the action $\Phi^{T^1_kQ}$, for a connected $G$.  In view of what we said before this means that $ \xi_Q^C(L)=0$  for all $\xi  \in \mathfrak{g}$.

Recall that, if $\{  E_a\}$ is  a member of a basis of $\mathfrak{g}$, we have used the notation $ \widetilde{E}_a=(E_a)_Q$ for its associated fundamental vector field on $Q$.

 \begin{prop} \label{Gammainv}
 The Lagrangian $k$-vector fields ${\mathbf \Gamma}$ of a regular invariant Lagrangian $L$ are $G$-invariant.
\end{prop}

\begin{proof}
Given that all the vector fields in the expression  $\Gamma_\alpha=v^A_\alpha Z_A^C+ (\Gamma_\alpha)_\beta^A Z^{V_\beta}_A$  are  invariant, and given that also the quasi-$k$-velocities are invariant functions, we only need to check that $[\widetilde{E}_a^C,\Gamma_{\alpha}]=0$. Since
\begin{equation}\label{Corchete}
[\widetilde{E}_a^C,\Gamma_{\alpha}]=\widetilde{E}_a^C((\Gamma_\alpha)_\beta^A)Z_A^{V_\beta},
\end{equation}
this will be the case if we can show that the functions $(\Gamma_\alpha)_\beta^A$ are invariant. When we apply the vector field $\widetilde{E}_a^C$ to both sides of the equations (\ref{elZ}), we may interchange  the derivatives $\widetilde{E}_a^C$ and $Z_A^C$, etc., because of their zero Lie brackets. One easily establishes that,  in view of $\widetilde{E}_a^C(L)=0$, what remains is
\[
[\widetilde{E}_b^C,\Gamma_\alpha](Z_B^{V_\alpha}(L))=0.
\]
By making use of expression (\ref{Corchete}), this is equivalent with
\[
\widetilde{E}_a^C((\Gamma_\alpha)_\beta^A)Z_A^{V_\beta}(Z_B^{V_\alpha}(L)) = 0.
\]
Given that the matrix $(Z_A^{V_\beta}(Z_B^{V_\alpha}(L)))$  has maximal rank for a regular Lagrangian (see Proposition~\ref{rank}), the result follows. 
\end{proof}

Since ${\mathbf\Gamma}$ is invariant, it reduces to a $k$-vector field $\breve{\mathbf\Gamma}$ on $(T^1_kQ)/G$. The goal of the next few sections is to provide a coordinate expression of this $k$-vector field. To do so, we will need to invoke a principal connection on the bundle $Q\to Q/G$.

\subsection{Local frames of vector fields on $T^1_kQ$} \label{sec43}

Suppose we are given a principal connection on the principal bundle $\pi_Q: Q \to Q/G$. In what follows, we want to re-express a Lagrangian $k$-vector field ${\mathbf\Gamma}$ in terms of a local frame $\{Z_A\}$ on $Q$.  We have now two choices to do so: either by making use of $\{X_i,{\widetilde E}_a\}$ (not an invariant frame) or by $\{X_i,{\widehat E}_a\}$ (an invariant frame), see Section~\ref{se52n} (with now $M=Q$). When we choose the frame $\{X_i,\widehat{E}_a\}$, we will write $(q^i,q^a,v^i_\alpha,w^a_\alpha)$ for the coordinates and the corresponding  quasi-$k$-velocities on $T^1_kQ$. If we use  $\{X_i, \widetilde{E}_a\} $, we will denote them as $(q^i,q^a,v^i_\alpha,v^a_\alpha)$.

From Lemma \ref{lem51}  we know that
the frame
 $\{Z_A^C,Z_A^{V_\alpha}\} =
\{X_i^C,X_i^{V_\alpha},\widehat{E}_a^C, \widehat{E}_a^{V_\alpha}\}$
 consists only of invariant vector fields on $\tkq$. Also the coordinate functions $q^i,v^i_\alpha,w^a_\alpha$ are $G$-invariant functions on $T^1_kQ$ and, therefore, they can be used as coordinates on
$(T^1_kQ)/G$.
In summary, we may say that the canonical projections are locally given by
$$
\begin{array}{ccc}
\pi_Q: Q & \to & Q/G \\ \noalign{\medskip}
(q^i,q^a) & \mapsto & (q^i)
\end{array}
\qquad  \quad
\begin{array}{ccc}
\pi_{T^1_kQ} : T^1_kQ & \to & (T^1_kQ)/G \\ \noalign{\medskip}
(q^i,q^a,v^i_\alpha,w^a_\alpha) & \mapsto & (q^i,v^i_\alpha,w^a_\alpha).
\end{array}
$$
\begin{lem}\label{lem28}
If we apply the vector fields $X_i^C, X_i^{V_\alpha}, \widehat{E}_a^C,
\widehat{E}_a^{V_\alpha}$ to the invariant functions $q^i, v^i_\alpha,
w^a_\alpha$ we obtain
 \[\begin{array}{lll}
  X_i^C(q^j)=\delta^j_i\, , &  X_i^C(v^j_\beta)=0 \, , &  X_i^C(w^b_\beta)=-\Upsilon^b_{ic}w^c_\beta+K^b_{ik}v^k_\beta \, , \\ \noalign{\medskip}
 X_i^{V_\alpha}(q^j)=0 \, , & X_i^{V_\alpha}(v^j_\beta)=\delta^j_i\delta^\alpha_\beta \, , & X_i^{V_\alpha}(w^b_\beta)=0\, , \\ \noalign{\medskip}
 \widehat{E}_a^C(q^j)=0\, , & \widehat{E}_a^C(v^j_\beta )=0\, , & \widehat{E}_a^C(w^b_\beta)=\Upsilon^b_{ka}v^k_\beta -C^b_{ac}w^c_\beta \, , \\ \noalign{\medskip}
 \widehat{E}_a^{V_\alpha }(q^j)=0\, , & \widehat{E}_a^{V_\alpha }(v^j_\beta )=0\, , & \widehat{E}_a^{V_\alpha }(w^b_\beta)=\delta^b_a \delta^\alpha_\beta \, ,
  \end{array}
  \]
\end{lem}
\begin{proof} Let $\{\vartheta^j,\varpi^a\}$ be the dual basis of $\{X_i,{\widehat E}_a\}$. From the bracket relations (\ref{e2}), we can see that ${\mathcal L}_{X_i}\vartheta^j=0$  and that ${\mathcal L}_{X_i}\varpi^b=-\Upsilon^b_{ic}\varpi^c+K^b_{ik}\vartheta^k$. Therefore,
$$
X_i^C(v^j_\beta)=X_i^C(\overrightarrow{\vartheta^j_\beta})= \overrightarrow{({\mathcal L}_{X_i}\vartheta^j)_\beta}=0
$$
and
$$
X_i^C(w_\beta^b)=X_i^C( \overrightarrow{\varpi^b_\beta})=
\overrightarrow{({\mathcal L}_{X_i}\varpi^b)_\beta}=-\Upsilon^b_{ic}w^c_\beta+K^b_{ik}v^k_\beta.
$$
Since also
$$
X_i^C(q^j)=X_i(q^j)=\delta^j_i,
$$
the first row in the Lemma follows.
The other properties follow in the same way.
\end{proof}

\begin{lem}\label{projvf}
The  projections of the $G$-invariant vector fields
$X_i^C,X_i^{V_\alpha},\widehat{E}_a^C, \widehat{E}_a^{V_\alpha}$ onto
$T^1_kQ/G$ are locally given, respectively, by
\[\begin{array}{ll}
 \breve{X}_i^C  =    \ds\frac{\partial}{\partial
q^i}+(K^b_{ik}v^k_\beta -\Upsilon^b_{ic} w^c_\beta)\ds\frac{\partial}
{\partial w^b_\beta },   & \qquad
 \breve{X}_i^{V_\alpha} =  \ds\frac{\partial}{\partial v^i_\alpha},
   \\
\breve{{E}}_a^{C}  =
(\Upsilon^b_{ka}v^k_\beta-C_{ac}^b w^c_\beta) \ds\frac{\partial}{\partial
w^b_\beta},    & \qquad
\breve{{E}}_a^{V_\alpha}  =  \ds\frac{\partial}{\partial w^a_\alpha}.
\end{array}
\]
\end{lem}
\begin{proof}
From the expressions in Lemma~\ref{lem28} and the relation (\ref{wf}) between an invariant vector field and its reduction, we obtain:
$$
\breve{X}_i^C(q^j )\circ \pi_{T^1_kQ}=X^C_i(q^j\circ \pi_{T^1_kQ})=X^C_i(q^j )=\delta^i_j,
$$
$$
\breve{X}_i^C(v^j_\beta )\circ \pi_{T^1_kQ}=X^C_i(v^j_\beta \circ \pi_{T^1_kQ})=X^C_i(v^j_\beta )=0,
$$
$$
\breve{X}_i^C(w^b_\beta )\circ \pi_{T^1_kQ}=X^C_i(w^b_\beta \circ \pi_{T^1_kQ})=X^C_i(w^b_\beta )= K^b_{ik}v^k_\beta -\Upsilon^b_{ic} w^c_\beta.
$$
Since $(q^i,v^i_\alpha,w^a_\alpha)$ forms a set of coordinate functions on $(T^1_kQ)/G$, this determines the vector field completely.
The same idea allows us to prove the other relations.
\end{proof}

\subsection{The reduced Lagrangian {\sc sopde}:   Lagrange-Poincar\'e field equations}

Assume that an invariant Lagrangian $L\in\cinfty{T^1_kQ}$ is given. Then, its derivatives by invariant vector fields, i.e.\ the functions $X_i^{C}(L),\, X_i^{V_\alpha}(L), \, \widehat{E}_a^{C}(L) ,\, \widehat{E}_a^{V_\alpha}(L)$, are invariant.
From relation (\ref{wf}), we can therefore write
\begin{equation}\label{ss}\begin{array}{l}
X_i^{V_\alpha}(L)=X_i^{V_\alpha}(l\circ \pi_{T^1_kQ})=\breve{X}^{V_\alpha}_i(l)\circ  \pi_{T^1_kQ}, \\ \noalign{\medskip}
X_i^{C}(L)=X_i^{C}(l\circ \pi_{T^1_kQ})=\breve{X}^{C}_i(l)\circ \pi_{T^1_kQ}, \\ \noalign{\medskip}
\widehat{E}_a^{V_\alpha}(L)=\widehat{E}_a^{V_\alpha}(l\circ \pi_{T^1_kQ})=\breve{E}_a^{V_\alpha}(l)\circ
  \pi_{T^1_kQ}, \\ \noalign{\medskip}
  \widehat{E}_a^{C}(L)=\widehat{E}_a^{C}(l\circ \pi_{T^1_kQ})=\breve{E}_a^{C}(l)\circ \pi_{T^1_kQ},
\end{array}\end{equation}
where $l:(T^1_kQ)/G \to \r$ is the {\sl reduced Lagrangian},  defined by $l\circ
\pi_{T^1_kQ} =L$.

From Proposition~\ref{frameprop}, we know that  {\sl Lagrangian {\sc sopde}}    ${\mathbf\Gamma}$ satisfies, with respect to the frame $\{X_i,{\widehat E}_a\}$, the equations
\[
\begin{array}{c}
  \Gamma_\alpha(X_i^{V_\alpha}(L))-
X_i^C(L)=0,  \\
  \Gamma_\alpha(\widehat{E}_a^{V_\alpha}(L))-
\widehat{E}_a^C(L)=0.\end{array}
\]
By making use of the fact that each  $\Gamma_\alpha$ is an invariant vector field on $T^1_kQ$, it follows from (\ref{wf}) and (\ref{ss})  that the reduced vector fields ${\breve\Gamma}_\alpha$ satisfy
\[
\begin{array}{l}
  \breve{\Gamma}_\alpha(\breve{X}_i^{V_\alpha}(l))-
\breve{X}_i^C(l)=0,  \\
  \breve{\Gamma}_\alpha( \breve{E}_a^{V_\alpha}(l))-
\breve{E}_a^C(l)=0,\end{array}
\]
 on $(T^1_kQ)/G$. Taking into account the result in Lemma
\ref{projvf} we can rewrite these equations as
\begin{equation}\label{eqelred}
\begin{array}{l}
  \breve{\Gamma}_\alpha\left(\ds\frac{\partial l} {\partial
v^i_\alpha }\right)- \ds\frac{\partial l} {\partial q^i}\, = \, \left(K^b_{ik}v^k_\beta -
\Upsilon^b_{ic} w^c_\beta \right)\ds\frac{\partial l} {\partial w^b_\beta},\\
\noalign{\medskip}   \breve{\Gamma}_\alpha\left(
\ds\frac{\partial l} {\partial w^a_\alpha }\right) \, = \,
\left(\Upsilon^b_{ka}v^k_\beta -C_{ac}^b w^c_\beta \right)
\ds\frac{\partial l} {\partial w^b_\beta}.
\end{array}
\end{equation}

A {\sc sopde} can be written as
\begin{equation}\label{gamlocl}
\Gamma_\alpha =v^i_\alpha X^C_i + w^a_\alpha \widehat{E}_a^C +
(\widehat{\Gamma}_\alpha )^j_\beta X_j^{ V_\beta}+ (\widehat{\Gamma}_\alpha )^a_\beta
\widehat{E}_a^{ V_\beta}.
\end{equation}
We have already established in the proof of Proposition~\ref{Gammainv} that the functions $(\widehat{\Gamma}_\alpha )^j_\beta$ and $(\widehat{\Gamma}_\alpha )^a_\beta$
are  invariant, and that they can be identified with functions on $(T^1_kQ)/G$. From Lemma~\ref{projvf}, we see that:

\begin{lem}\label{ivfgam}
 The reduced vector fields $\breve{\Gamma}_\alpha $ on $(T^1_kQ)/G$ of
  a {\sl Lagrangian {\sc sopde}} $\Gamma_\alpha$ are given by
$$
\begin{array}{rcl}
\breve{\Gamma}_\alpha & = & v^i_\alpha\ds\frac{\partial} {\partial
q^i}+(\widehat{\Gamma}_\alpha)^j_\beta \ds\frac{\partial} {\partial v^j_\beta}
+ (\widehat{\Gamma}_\alpha)^c_\beta \ds\frac{\partial} {\partial w^c_\beta} \\
\noalign{\medskip}
  & & + \, \left( \Upsilon^c_{ib} (v^i_\beta w^b_\alpha -  v^i_\alpha w^b_\beta)
 -C^c_{ab} w^a_\alpha w^b_\beta + K^c_{ij} v^i_\alpha v^j_\beta \right) \ds\frac{\partial}{\partial w^c_\beta}
\end{array}
$$
\end{lem}

Then, if  ${\breve\phi}(t)=(q^i=\phi^i(t),v^i_\alpha=\phi^i_\alpha(t),w^a_\alpha=\phi^a_\alpha(t))$ is an integral
section of $\breve{\mathbf\Gamma}$ then it satisfies,
in view of relations (\ref{eqelred}), the {\sl Lagrange-Poincar\'e field equations},
\begin{equation}\label{l-eq}
\begin{array}{ll}
\ds\fpd{\phi^i}{t^\alpha} \,=\, \ds\phi^i_\alpha,\\\noalign{\medskip}
\ds\frac{\partial}{\partial t^\alpha} \left(
\ds\frac{\partial l}{\partial v^i_\alpha}  \right) \, - \,
\ds\frac{\partial l}{\partial q^i} \, = \,
\left(K^b_{ik} \phi^k_\beta
-\Upsilon^b_{ic} \phi^c_\beta \right)\ds\frac{\partial l} {\partial w^b_\beta},\\
\noalign{\medskip}  \, \ds\frac{\partial}{\partial
t^\alpha} \left( \ds\frac{\partial l}{\partial w^a_\alpha}
\right) \, = \, \left(\Upsilon^b_{ka} \phi^k_\beta -C_{ac}^b \phi^c_\beta \right)
\ds\frac{\partial l} {\partial w^b_\beta}.
\end{array}
\end{equation}

 The equations above agree with the Lagrange-Poincar\'e equations as they appear in \cite{Ellis}, if one takes two issues into account. The first is that the current setting (the $k$-symplectic formalism) is different from the one in \cite{Ellis} (a jet bundle formalism). One way to relate the two approaches is by choosing the base space of the jet bundle to be simply $\r^k \times Q$, and to assume that the Lagrangian does not depend explicitly on the parameters $t^\alpha$. The second issue is that only coordinate-independent expressions appear in \cite{Ellis}, at the price of assuming to have an extra covariant derivative at disposal. This covariant derivative is actually only required to give a geometric sound meaning to all the separate terms in the equations, but it disappears from the equations when one only considers their coordinate expressions. This observation is already apparent when one considers only the simplest case of Lagrangian mechanics (with $k=1$ in the current setting), see e.g.\ the first remark on page 35 of the booklet \cite{Cendra}. If one takes the above remarks into account, and if one calculates coordinate expressions of the Lagrange-Poncar\'e equations as they appear in \cite{Ellis}, the two sets of expressions compare. In the special case where the configuration space $Q$ coincides with the symmetry group $G$, the equations simplify to  equations on $(T^1_kG)/G={\g}^k$, the so-called {\sl Euler-Poincar\'e field equations} given by
\[   \frac{\partial}{\partial
t^\alpha} \left( \ds\frac{\partial l}{\partial w^a_\alpha} \right)
\, = \,  -C_{ac}^b \phi^c_B \  \ds\frac{\partial l} {\partial
w^b_B}.
\]
These equations agree with those in \cite{CRS}, when one considers coordinate expressions.

We have established, in view of Proposition~\ref{m} that a solution of the Euler-Lagrange equations (\ref{lfield}) projects onto a solution of the Lagrange-Poincar\'e equations (\ref{l-eq}). However, we can not conclude that any solution of (\ref{l-eq}) can be extended to one of (\ref{lfield}).
For that reason, we need to study the integrability conditions of the Lagrangian $k$-vector fields ${\mathbf \Gamma}$.

\subsection{The integrability of an invariant {\sc sopde}} \label{sec53}

We now specify the results of Section~\ref{sec5} to the case where the $k$-vector field ${\mathbf X}$ is a {\sc sopde} ${\mathbf \Gamma}$ on $M=T^1_kQ$. We will also draw an analogy with some results of the paper \cite{Ellis}, when  translated to the current framework.

Recall that we are working with the lifted action of $G$ on $M=T^1_kQ$. As we had done in Section~\ref{se52n}, we may introduce the vector fields
\[
{\widehat E}_a^M = A^b_a  {\widetilde E}_b^C,
\]
on $T^1_kQ$. They correspond with the invariant vector fields we had introduced in expression (\ref{Ehat}) (or (\ref{A})), but now for the lifted $G$-action on $M=T^1_kQ$.
Given that ${\widehat E}_a = A^b_a  {\widetilde E}_b$, we get
\[
{\widehat E}^C_a = A^b_a  {\widetilde E}^C_b + \fpd{{A}^b_a}{t^\beta} {\widetilde E}^{V_\beta}_b =   {\widehat E}^M_b + ( X_i(A^b_a)v^i_\beta +{\widehat E}_c(A^b_a)w^c_\beta ) {\widetilde E}^{V_\beta}_b =  {\widehat E}_a^M + (  \Upsilon^b_{ia}v^i_\beta + C^b_{ca} w^c_\beta ) {\widehat E}^{V_\beta}_b.
\]

To proceed as in Section~\ref{se52n} we need a principal connection on the bundle $\pi_{T^1_kQ}:M=T^1_kQ \to M/G=(T^1_kQ)/G$. It will be most convenient to define this connection by means of its connection map that takes values in the Lie algebra $\g$.
\begin{definition}
Let $\vartheta^Q: TQ \to \g$ be a principal connection on $\pi_Q:Q\to Q/G$, then its vertical lift  $\vartheta^{T^1_kQ}: T(T^1_kQ) \to \g$ is the principal connection on $\pi_{T^1_kQ}$, given by
\[
\vartheta^{T^1_kQ} (W) = \vartheta^Q (T\tau^1_Q(W)),
\]
for all $W\in T(T^1_kQ)$, where $\tau^1_Q$ is the natural projection $: T^1_kQ \to Q$.
\end{definition}
The fact that this connection is principal, follows easily from the fact that the connection $\omega^Q$ is, and from the property $T\tau^1_Q(gW)=g(T\tau^1_Q(W))$.


We will denote the corresponding connection map by $\omega^{T^1_kQ}: \vectorfields{T^1_kQ} \to \vectorfields{T^1_kQ}$. Its relation to $\omega^Q: \vectorfields{Q} \to \vectorfields{Q}$ is
\[
\omega(W)({\mathbf v}) = (\vartheta(W({\mathbf v})))^C_Q({\mathbf v}),
\] for all $W\in\vectorfields{T^1_kQ}$ and ${\mathbf v}\in T^1_kQ$.  From $\omega^{T^1_kQ} ({\widetilde E}_a^C) = {\widetilde E}_a^C$ it follows that the action of  $\omega^{T^1_kQ}$ on the invariant frame $\{{\widehat E}_a^C,X_i^C, {\widehat E}_a^{V_\beta},X_i^{V_\beta}\}$  is given by
\begin{eqnarray*}
&& \omega^{T^1_kQ} ({\widehat E}_a^C) = {\widehat E}_a^M = {\widehat E}^C_a - (  \Upsilon^b_{ia}v^i_\beta + C^b_{ca} w^c_\beta ) {\widehat E}^{V_\beta}_b,
\\ && \omega^{T^1_kQ} ({\widehat E}_a^{V_\beta}) = 0, \qquad  \omega^{T^1_kQ} (X_i^C) = 0,\qquad  \omega^{T^1_kQ} (X_i^{V_\beta}) = 0.
\end{eqnarray*}

For later use we give the decomposition of  a {\sc sopde} ${\bf\Gamma}$ in its horizontal and vertical part, with respect to the connection $\omega^{T^1_kQ}$. If we write  $\Gamma_\alpha$ as in expression (\ref{gamlocl}), then $\omega^{T^1_kQ}(\Gamma_\alpha) =   w^a_\alpha  {\widehat E}^M_a$ and the horizontal part of $\Gamma_\alpha$ is
$ v^i_\alpha X^C_i +
(\widehat{\Gamma}_\alpha )^j_\beta X_j^{ V_\beta}+ (\widehat{\Gamma}_\alpha )^a_\beta
\widehat{E}_a^{ V_\beta}     + (  \Upsilon^b_{ia}v^i_\beta w^a_\alpha + C^b_{ca} w^c_\beta w^a_\alpha) {\widehat E}^{V_\beta}_b.$
When we compute the reduced vector field of this horizontal part, using the expressions of Lemma~\ref{projvf}, it, of course, coincides with the expression of the reduced vector fields ${\breve\Gamma}_\alpha$ we had obtained in Lemma~\ref{ivfgam}. The decomposition (\ref{dec3}) on $M=T^1_kQ$ for  $\Gamma_\alpha$, using the  connection  $\omega^{T^1_kQ}$, is then:
\begin{equation} \label{decnew}
\Gamma_\alpha = (\breve \Gamma_\alpha)^h + ({\bar\Gamma}_\alpha)^v = (\breve \Gamma_\alpha)^h + w^a_\alpha  {\widehat E}^M_a,
\end{equation}
that is to say: in the notations of Section~\ref{se52n} the section ${\bar\Gamma}_\alpha$ $\in Sec(\bar\g \to (T^1_kQ)/G)$ has coefficients ${X}^a_\alpha =w^a_\alpha$.

 We know from Proposition~\ref{Gammainv} that a $k$-vector field ${\mathbf\Gamma}$ of an invariant Lagrangian $L$ is $G$-invariant on $T^1_kQ$.
Next to  requiring the integrability of $\check{\mathbf \Gamma}$, the integrability of ${\mathbf \Gamma}$ is guaranteed if the coordinate expression  (\ref{seccond}) is satisfied. In it we need the curvature  of the vertical lift connection $\omega^{T^1_kQ}$.

From the defining relation, and from the coordinate expressions, it is clear that $\omega^{T^1_kQ}$ has the property that, if the vector fields $W\in\vectorfields{T^1_kQ}$ and $X\in\vectorfields{Q}$ are $\tau^1_Q$-related, then so are the vector fields $\omega^{T^1_kQ}(W)$ and $\omega^Q(X)$. There exists a similar property for the horizontal lifts that correspond to each of the two connections. We will denote the horizontal lift of $\omega^Q$ by $h$, and the horizontal lift of $\omega^{T^1_kQ}$ by $h_k$. We will also use the notation ${\tilde\tau}^1_Q$ for the projection $M/G=(T^1_kQ)/G \to Q/G$. It easily follows that, if $\breve W \in \vectorfields{M/G}$ and $\breve X \in \vectorfields{Q/G}$ are ${\tilde\tau}^1_Q$-related, then ${\breve W}^{h_k}$ and ${\breve X}^h$ are $\tau^1_Q$-related. For two pairs of such vector fields, it follows that $[{\breve W}^{h_k}_1,{\breve W}^{h_k}_2]$  is $\tau^1_Q$-related to $[{\breve X}^h_1,{\breve X}^h_2]$. From all this, we may conclude that that the curvatures $({\bar K}^{T^1_kQ}({\breve W}_1,{\breve W}_2))^v = -\omega^{T^1_kQ}[{\breve W}_1^{h_k},{\breve W}_2^{h_k}] )$ and $({\bar K}^{Q}({\breve X}_1,{\breve X}_2))^v = -\omega^{Q}[{\breve X}_1^{h},{\breve X}_2^{h}])$ are  $\tau^1_Q$-related whenever the arguments are ${\tilde\tau}^1_Q$-related. Here $()^v$ stands for either the vertical lift associated to the fibre bundle $T^1_kQ\to (T^1_kQ)/G$ or to the bundle $Q\to Q/G$.  From this property we may deduce that the only non-vanishing curvature coefficients of $K^{T^1_kQ}$ are actually those of $K^Q$.


Likewise, for the adjoint connection, if $\breve W$ is a vector field on $(T^1_kQ)/G$ that is ${\tilde\tau}^1_Q$-related to a vector field $\breve X$ on $Q/G$, and if $\bar Z$ is a section of $(T^1_kQ\times \g)/G$ that is related to a section $\bar Y$ of $(Q\times \g)/G$, then one may show that $(\nabla^{T^1_kQ}_{\breve W}\bar Z)^v = [\breve{W}^{h_k},\bar Z^v]$ is $\tau^1_Q$-related to
$(\nabla^{Q}_{\breve X}\bar Y)^v = [\breve{X}^{h},\bar Y^v]$. Again, in terms of the connection coefficients of the connection $\nabla^{T^1_kQ}$, this means that the only connection coefficients that matter are those of $\nabla^Q$.


We can now easily compute the coordinate expression (\ref{seccond}), for the case ${\mathbf X} = {\mathbf \Gamma}$. We reach the following conclusion:
\begin{prop}  A {\sc sopde} ${\mathbf \Gamma}$ is integrable, if and only if its reduced $k$-vector field  $\breve{\mathbf \Gamma}$ is, and if
 \[
{\breve\Gamma}_\alpha(w^b_\beta)  - {\breve\Gamma}_\beta(w^b_\alpha) + (v^i_{\alpha} w^a_\beta   - v^i_{\beta} w^a_\alpha)   \Upsilon_{ia}^{b}   + C^b_{ac}w^a_\alpha w^c_\beta  - K^b_{ij}v^i_\alpha v^j_\beta  =0.
\]
\end{prop}

In terms of the integral curves $(q^i=\phi^i(t),v^i_\alpha = \phi^i_\alpha(t),w^a_\alpha = \phi^a_\alpha(t))$ of the reduced $k$-vector field $\breve{\mathbf \Gamma}$ this means that
\[
\fpd{\phi^b_\beta}{t^\alpha}  - \fpd{\phi^b_\alpha}{t^\beta} + (\phi^i_{\alpha} \phi^a_\beta   - \phi^i_{\beta} \phi^a_\alpha)   \Upsilon_{ia}^{b}   + C^b_{ac}\phi^a_\alpha \phi^c_\beta  - K^b_{ij}\phi^i_\alpha \phi^j_\beta  =0.
\]
This condition represents the analogue of expression (3.29) of \cite{Ellis} in our formalism. Since we are also assuming that the reduced $k$-vector field $\breve{\mathbf \Gamma}$ is integrable, i.e.\ $[{\breve\Gamma}_\alpha,{\breve\Gamma}_\beta] =0$, we find that, among other,  the integral curves satisfy
\[
\fpd{\phi^i_\beta}{t^\alpha}  - \fpd{\phi^i_\alpha}{t^\beta}=0.
\]

When $Q=G$, we simply get
 \[
{\breve\Gamma}_\alpha(w^b_\beta)  - {\breve\Gamma}_\beta(w^b_\alpha)
    + C^b_{ac}w^a_\alpha w^b_\beta   =0,
\]
which is our analogue of the condition about vanishing curvature in Theorem~3.2 of \cite{CRS}.  If we use the vertical lift $\omega^{T^1_kQ}$ to define the connections $\omega^{\breve\phi}$ and $\omega^{\breve\phi, \bar{\mathbf \Gamma}}$ that appear in the expressions (\ref{con2}) and (\ref{con3}) of Section~\ref{se52n}, we get our analogues of the connections  ${\mathcal A}^\rho$ and ${\mathcal A}^{\bar\sigma}$ that appear in the paper \cite{Ellis} in their section on `Reconstruction conditions'.

\section{Reconstruction} \label{sec6}

The integrability conditions we have discussed so far only give necessary and sufficient  conditions for an integral section of the invariant vector field to exist. They do, however, not provide a method by which one can actually construct such a section. In this section, we provide such a method in Section~\ref{sec62}. First we need to define the notion of a $k$-connection.

\subsection{$k$-connections and principal $k$-connections} \label{sec61}

Consider a fibre bundle $\pi: M\to N$, with local adapted coordinates $(x^i,x^a)$. In this section we introduce the notion of a $k$-connection on $\pi: M \to N$. We can extend the short exact sequence (\ref{seq}) to the level of $T^1_kM$, as follows:
\[
0 \to (VM)^k \to T^1_kM \to M\times_N T^1_kN \to 0.
\]
The middle arrow is now given by $j^k: T^1_kM \to M\times_N T^1_kN: {\mathbf v} \mapsto (\tau({\mathbf v}),T^1_k\pi({\mathbf v}))$; its kernel is given by $(VM)^k$ ($k$ copies of $VM$).

\begin{definition}
A $k$-connection on $\pi: M\to N$ is a linear bundle map $\gamma^k:M\times_N T^1_kN \to T^1_k M$ which is such that $j^k\circ\gamma^k =id$.
\end{definition}

Locally, $\gamma^k$ will be of the form $\gamma^k: (x^i;x^a, u^i_\alpha) = (x^i,x^a, u^i_\alpha , u^a_\alpha = -B^{a\beta}_{i\alpha} u^i_\beta)$, for some `connection coefficients' $B^{a\beta}_{i\alpha} \in \cinfty{M}$. We will denote the corresponding right splitting, thought of as a (1,1) $k$-tensor field on $M$, by $\omega^k:  \mathfrak{X}^k(M) \to \mathfrak{X}^k(M) $. Any $k$-vector field ${\mathbf X}$ on $M$ can be decomposed into a horizontal part ${\mathbf X}-\omega^k({\mathbf X})$ and a vertical part $\omega^k({\mathbf X})$.

Given a $k$-vector field $\mathbf Y$ on $N$ we can define its horizontal lift as the $k$-vector field ${\mathbf Y}^H$ on $M$, given by
\[
{\mathbf Y}^H(m) = \gamma^k(m, {\mathbf Y}(\pi(m))).
\]
 If $Y_\alpha = Y^i_\alpha \partial / \partial x^i$, we get that $({\mathbf Y}^H)_\alpha = Y^i_\beta X_{i\alpha}^\beta$, where, from now on, we will use the notation
 \begin{equation} \label{horbasis}X_{i\alpha}^\beta=\delta^\beta_\alpha \frac{\partial}{\partial x^i} - B^{a\beta}_{i\alpha} \frac{\partial}{\partial x^a} \in\vectorfields{M}.
 \end{equation}
We now give two examples of $k$-connections. A third example, what we have called 'the mechanical connection', is given in Section~\ref{sec63}.

{\bf Example 1.\ A `simple' connection.} It is easy to see that we can construct a $k$-connection from a genuine connection $\gamma^M$ on $\pi:M\to N$, given by
\[
\gamma^k(m, {\mathbf u}) = (\gamma^M(m, u_1), \ldots, \gamma^M(m, u_k)  ), \quad  {\mathbf u}\in T^1_kN.
\]
The map $\gamma^M$ is locally given by $\gamma^M(x^i; x^a, {\dot x}^i) = (x^i,x^a, {\dot x}^i, {\dot x}^a= \Gamma^a_i(x) {\dot x}^i)$, for some connection coefficients $\Gamma^a_i$.
In this case, $B^{a\beta}_{i\alpha} =  \Gamma^{a}_{i} \delta^\beta_\alpha$. We will often refer to this kind of $k$-connections as those of `simple' type.

{\bf Example 2.\ The {\sc sopde} connection.} Take $M=T^1_kQ$ and $N=Q$, and assume that ${\mathbf \Gamma} = (\Gamma_\alpha)$ is a {\sc sopde}.  Denote by ${\mathbf S}^\gamma$ the (1,1) $k$-tensor field, given by
\[
({\mathbf S}^\gamma ({\mathbf X}))_\beta = S^\gamma (X_\beta) = X^A_\beta \fpd{}{u^A_\gamma}.
 \]
 Then, with the definition of the Lie derivative (\ref{lieder}) as we defined in Section~\ref{sec2}, ${\mathcal L}_{\Gamma_\gamma} {\mathbf S}^\gamma$ (sum over $\gamma$) is again a (1,1) $k$-tensor field on $M=T^1_kQ$. We can define a $k$-connection on the bundle $T^1_kQ \to Q$ by saying that its connection map (i.e.\ its vertical projector) is
\[
\omega^k = \frac{1}{k+1}\left( k\, Id + {\mathcal L}_{\Gamma_\gamma} {\mathbf S}^\gamma \right).
 \]
 It is easy to see that, for a {\sc sopde} with $\Gamma_\beta = u_\beta^A \partial / \partial{q^A} + (\Gamma_{\beta})_\alpha^A \partial / \partial{u^A_\alpha}$ and a $k$-vector field with $X_\beta = X_\beta^A \partial / \partial{q^A} + (X_{\beta})_\alpha^A \partial / \partial{u^A_\alpha}$, the $\beta$th vector field of the $k$-vector field $\omega^k({\mathbf X})$ is given by
\[
(\omega^k({\mathbf X}))_\beta = ((X_{\beta})_\gamma^A + X^C_\alpha B^{A\beta}_{C\gamma\alpha} ) \fpd{}{u^A_\gamma}
\]
where the connection coefficients are given by $B^{A\beta}_{C\gamma\alpha} = \delta^\beta_\alpha \Gamma^A_{C\gamma}$ and
\[
\Gamma^A_{C\gamma} = - \frac{1}{k+1} \fpd{}{u^C_\delta} \left( \Gamma_\delta\right)^A_{\gamma} .
 \]
 As the form of the connection coefficients $B^{A\beta}_{C\gamma\alpha}$ suggests, this $k$-connection is in fact of simple type. It is actually the one associated to the (genuine) connection on the bundle $T^1_k Q \to Q$ that was defined in the paper \cite{RRSV} for ${\mathbf\Gamma}$.

In the special case that the fibre bundle $\pi$ is a principal bundle $\pi_M: M \to N=M/G$, we can also define principal $k$-connections. In that case, we may identify the vertical distribution $VM$ with $M\times \g$ through $(\xi_M(m)) \mapsto (m,\xi)$. Given $\xi_\alpha$ in ${\g}$, we can define the {\em fundamental $k$-vector field} as $(\xi_1, \ldots, \xi_k)_M := ((\xi_1)_M,\ldots, (\xi_k)_M) \in {\mathfrak X}^k(M)$. We may also identify $(VM)^k$ with $M\times\g^k$, so that the short exact sequence of interest is given by
\[
0 \to M\times {\mathfrak g}^k \to T^1_kM \to M\times_{M/G} T^1_k(M/G) \to 0.
\]

Given a splitting $\gamma^k$ of this sequence, we can define a form $\vartheta^k: T^1_k M \to \g^k$, as the map which has the property that $\omega^k({\mathbf v}_m) = (\vartheta^k({\mathbf v}_m))_M(m)$. Then $\vartheta^k ((\xi_1, \ldots, \xi_k)_M) = (\xi_1, \ldots, \xi_k)$.

\begin{definition}
A $k$-connection $\gamma^k$ on $\pi_M: M\to M/G$ is  principal if
\[
\vartheta^k (g {\mathbf v}_m) = (Ad_{g^{-1}})^k (\vartheta^k({\mathbf v}_m)),
\]
where $(Ad_{g^{-1}})^k: \g^k \to \g^k$ is the application of  $Ad_{g^{-1}}: \g \to \g$ to each of the $k$ factors.
\end{definition}
 When expressed in terms of the (1,1) $k$-tensor field $\omega^k: T^1_kM \to T^1_kM$, the condition in the definition means that $\omega^k (g {\mathbf v}_m) = g \omega^k ({\mathbf v}_m)$. In view of the definition of the Lie derivative we had given in Section~\ref{sec2}, this is equivalent (when $G$ is connected) with ${\mathcal L}_{\xi_M}\omega^k = {\mathbf 0}$, when we consider the action of $\omega$ on $k$-vector fields.

 Likewise, we have for a principal connection that $\gamma^k({\mathbf u}_n, gm) = g \gamma^k({\mathbf u}_n, m)$. Assume that $\breve{\mathbf X}$ is a given $k$-vector field on $N$. In view of the previous property its horizontal lift will satisfy  ${\breve{\mathbf X}}^H (gm) = g {\breve{\mathbf X}}^H (m)$. The $k$-vector field ${\breve{\mathbf X}}^H$ on $M$ is thus always  invariant, meaning that ${\mathcal L}_{\xi_M}{\breve{\mathbf X}}^H ={\mathbf 0}$, for all $\xi\in\g$. In coordinates, this means that the vector fields $X_{i\alpha}^\beta$ on $M$ are all invariant, i.e.\ $[X_{i\alpha}^\beta, {\widetilde E}_a] = 0$.

We briefly say a few words about the integrability of a horizontal lift. Let $\breve{\mathbf X}$ be a given integrable $k$-vector field on $M/G$. For the special case with ${\mathbf X} = {\breve{\mathbf X}}^H$, Proposition~\ref{propint} tells us that if ${\mathbf X} = {\breve{\mathbf X}}^H$ is integrable then  the curvature of $\omega^{\breve\phi, {\breve{\mathbf X}}^H }$  should also vanishes.

From (\ref{Vp}) and (\ref{horbasis}) we can write
   $$T_pi(V_p)= T^\alpha \partial/\partial {t^\alpha}\circ i(p) + (X^i_\beta\circ\breve\phi)  T^\alpha X_{i\alpha}^\beta \circ i(p)+ {Z}^a {\widehat E}_a \circ i(p)$$
where $Z^a$ is given by
$$Z^cA^b_cK^a_b=  (X^i_\beta\circ\breve\phi)  T^\alpha B_{i\alpha}^{a\beta} +\tilde{Y}^a  \, .$$
Thus,
from  (\ref{conmapphi}) and (\ref{conmapphi1}),  we obtain
\[ \omega^{\breve\phi, {\breve{\mathbf X}}^H } (V_{p}) =\left(t,({\tilde Y}^a - (\breve{\mathbf X}^H )^a_\alpha(m)  T^\alpha) \fpd{}{x^a}\Big|_{m} \right)=
 (t,{Z}^a {\widehat E}_a) \, .
\]

Let us now restrict our attention to the `simple case', when the $k$-connection $\gamma^k$ is constructed from a genuine connection $\gamma^M$. In that case, it is easy to give a second interpretation of the integrability conditions,  as we did in Proposition~\ref{newprop}. The horizontal lift of a $k$-vector field $\breve{\mathbf X}$ on $M/G$ is now of the form
$({\breve{\mathbf X}}^H)_\alpha = X^i_\alpha({\partial}/{\partial x^i} -  \Gamma^{a}_{i}{\partial}/{\partial x^a}  ) = (X_\alpha)^h$, where the last ${}^h$ stands for the horizontal lift associated to $\gamma^M$. We then know from (\ref{dec2}) that
\[
[({\breve{\mathbf X}}^H)_\alpha,({\breve{\mathbf X}}^H)_\beta] = [(\breve{X}_\alpha)^h,(\breve{X}_\beta)^h] = [\breve{X}_\alpha,\breve{X}_\beta]^h - (K^M(\breve{X}_\alpha,\breve{X}_\beta))^v.
\]
Here $K^M$ stands, as before, for the curvature of $\gamma^M$, taking values in the vertical distribution of $\pi: M\to M/G$. We can therefore conclude that
\begin{prop}
The horizontal lift ${\breve{\mathbf X}}^H$ corresponding to a simple $k$-connection is integrable if and only if $\breve{\mathbf X}$ is integrable and $K^M(\breve{X}_\alpha,\breve{X}_\beta)=0$, for all choices of $\alpha$ and $\beta$.
\end{prop}

\subsection{Reconstruction method} \label{sec62}

We will suppose throughout this section that $\Phi$ defines a free and proper action, and we will denote $\pi^M$ for the projection $M\to N= M/G$.
We will also assume that we have a principal $k$-connection $\gamma^k$ (or $\omega^k: T^1_k M \to {\mathfrak g}^k$)  at our disposal and we will assume that  $\breve{\mathbf X}$ and ${\breve{\mathbf X}}^H$ are both integrable.

Let $\breve\phi$ be a given integral section of the reduced vector field $\breve{\mathbf X}$.
\begin{definition}
A map ${\breve\phi}_H:\rk\to M$ is called a horizontal lift of $\breve\phi$ if (1) $\pi\circ{\breve\phi}_H={\breve\phi}$ and (2) ${\breve\phi}_H$ is an integral section of  ${\breve{\mathbf X}}^H$. \end{definition}
In local coordinates,  we denote $\breve\phi(t)=(x^i = \phi^i(t))$, and ${\breve\phi}_H(t)=(\phi^i(t),\phi_H^a(t))$. Using  (\ref{horbasis}) and that  $\breve\phi$ is an   integral section of the reduced vector field $\breve{\mathbf X}$ we obtain
\begin{equation}\label{horlift}
\frac{\partial  \phi_H^a}{\partial t^\alpha} = - B^{a\beta}_{i\alpha}  \frac{\partial  \phi^i}{\partial t^\beta}.
\end{equation}
This relation is equivalent with $\omega^k({\breve\phi}^{(1)}_H)=0$, where ${\breve\phi}^{(1)}_H$ stands for the first prolongation of ${\breve\phi}_H$ (see Section~\ref{sec2}), since ${\breve\phi}^{(1)}_H(t)=\breve{\mathbf X}^H({\breve\phi}_H(t))$.

Assume now given a map $g:\rk\to G$, then
$$\begin{array}{cccl}
g^{(1)}: & \rk & \to & T^1_kG \\ \noalign{\medskip}
              &   t  & \to  & g^{(1)}(t)=(\ldots,T_tg(\derpar{}{t^\alpha}\Big\vert_{t}),\ldots)
\end{array}$$
and for each $\alpha$
$$T_{g(t) }L_{ g^{-1}(t) }\left(   T_tg(\derpar{}{t^\alpha}\Big\vert_{t})\right)\in \g\, .$$

We denote by $g^{-1}(t) g^{(1)}(t)$ the element on $\g^k$ defined by
$$
\left(  T_{g(t) }L_{ g^{-1}(t) }\left(   T_tg(\derpar{}{t^1}\Big\vert_{t})\right), \ldots, T_{g(t) }L_{ g^{-1}(t) }\left(   T_tg(\derpar{}{t^k}\Big\vert_{t})\right)\right).
$$

\begin{lem}
When two maps $\phi,\psi: \r^k \to M$ are related by $\phi(t)=g(t)\psi(t)$, for some $g:\rk\to G$,  their prolongations satisfy
\begin{equation}\label{phi1}
\phi^{(1)} = g \left(\psi^{(1)} + (g^{-1} g^{(1)})_M \circ \psi\right).
\end{equation}
 which means that
$$
\phi^{(1)}_\alpha(t) = T_{\psi(t)}\Phi_{g(t)} \left[  \psi^{(1)}_\alpha(t) +(T_{g(t) }L_{ g^{-1}(t) }\left(    g^{(1)}_\alpha(t) )\right)_M(\psi(t) ) \right]\, .
$$
\end{lem}
\begin{proof}
The following property is well-known (see e.g.\ \cite{AM}). Let $v_h\in T_hG$ and $m\in M$. Set $\eta= h^{-1}v_h \in {\mathfrak g}$. Then
\[
T_h\Phi_m (v_h) = T_m \Phi_h (\eta_M(m)).
\]
By using the Leibniz rule and the above property, we obtain
\begin{eqnarray*}
\phi^{(1)}_\alpha (t)&=&
T_t\phi \left(\frac{\partial}{\partial t^\alpha}\Big|_t\right) = T_{\psi(t)} \Phi_{g(t)} \left(T_t\psi\left(\frac{\partial}{\partial t^\alpha}\Big|_t \right)\right) + T_{g(t)} \Phi_{\psi(t)} \left(T_t g\left(\frac{\partial}{\partial t^\alpha}\Big|_t \right)\right)
\\
 &=& T_{\psi(t)} \Phi_{g(t)} \left( T_t\psi\left(\frac{\partial}{\partial t^\alpha}\Big|_t \right) + (\xi_\alpha)_M \circ\psi\right).
\end{eqnarray*}
 Here, $\xi_\alpha$ stands for  $T_{g(t) }L_{ g^{-1}(t) }(   T_tg(\derpar{}{t^\alpha}\Big\vert_{t}))$, the  $\alpha$th component of $g^{-1} g^{(1)}: \rk \to {\mathfrak g}^k$. All the components together therefore lead to the desired property.
\end{proof}

The reconstruction problem is the following one. What are the conditions on $g(t)$ such that $\phi(t) = g(t)\breve{\phi}_H(t)$ is an integral section of $\mathbf X$? For that to be true, we must have that:
\[
\phi^{(1)} = {\mathbf X} \circ \phi
\]
or, in view of the property (\ref{phi1}), the invariance of ${\mathbf X}$ and the freeness of the action,
\begin{equation}\label{Ellisexpr}
{\breve\phi}_H^{(1)} + (g^{-1} g^{(1)})_M \circ {\breve\phi}_H = {\mathbf X} \circ {\breve\phi}_H.
\end{equation}
After applying the connection form $\omega^k$ on both sides we get that $g(t)$ must satisfy
\begin{equation} \label{receq}
g^{-1} g^{(1)}= \omega^k ({\mathbf X} \circ {\breve\phi}_H).
\end{equation}
This PDE in $g$ will be called the {\sl reconstruction equation}. If it has a solution $g(t)$, an integral section $\phi(t)$ for ${\mathbf X}$ may be reassembled from an integral section $\breve\phi(t)$ of $\breve{\mathbf X}$. We have shown:

\begin{prop}
Let ${\mathbf X}$ be an integrable and invariant $k$-vector field on $T ^1_kQ$ with integrable reduced $k$-vector field  $\breve{\mathbf X}$.
Let  $\breve{\phi}$ be an integral section of $\breve{\mathbf X}$
and $\breve{\phi}_H\colon \r^k\to M$ a horizontal lift of $\breve{\phi}$.
  If $g \colon \r^k \to G$ is a solution to the reconstruction equation (\ref{receq}), then $\phi\colon \r^k\to M$ defined by
$$\phi(t) = g(t)\breve{\phi}_H(t)$$
is an integral section of ${\mathbf X}$.\end{prop}

In the next section we will consider again the case where ${\mathbf X} = {\mathbf \Gamma}$ is given by Lagrangian field equations. For completeness, we mention that there are also reconstruction equations in the formalism of the paper \cite{Ellis}. Their reconstruction PDE in expression (3.32) can best be compared with our expression (\ref{Ellisexpr}).

\subsection{The mechanical $k$-connection} \label{sec63} With respect to the notations of the previous paragraphs we  take again $M=T^1_kQ$ and $N=(T^1_kQ)/G$ and ${\mathbf X} = {\mathbf \Gamma}$ a Lagrangian {\sc sopde}. In order for the reconstruction method to work, we need a  $k$-connection on $\pi_{T^1_kQ}$. We now show how to construct one from the given Lagrangian.

The vertical space $V^k$ in the short exact sequence
\[
0 \to V^k \to T^1_k (T^1_kQ) \to T^1_kQ\times_{T^1_kQ/G} T^1_k(T^1_kQ/G) \to 0
\]
can now be identified with  $T^1_kQ \times {\mathfrak g}^k$. Let ${\mathbf v} = (q;v_1,\ldots v_k)\in T^1_k Q$ be such that $\tau^1_Q({\mathbf v}) =q$, where $\tau^1_Q: T^1_k Q \to Q$.
The set of vertical elements is spanned by elements of the form $\xi^a_\alpha {\widetilde E}^C_a ({\mathbf v})$, i.e.\ couples of the type $(\xi_1,\ldots,\xi_k)_{T^1_kQ}({\mathbf v}) $.

 Given an invariant Lagrangian $L\in\cinfty{T^1_kQ}$, we will show below how to define a splitting $\gamma^k$ (or, equivalently $\omega^k$) of this sequence, under a certain regularity assumption for the Lagrangian.

Consider the $k$-symplectic forms $\omega^\alpha_L$ of $L$. We define linear maps
\[ \begin{array}{ccccl}
g_{{\mathbf v}}^{\alpha,\beta} & : & T_qQ \times    T_qQ & \to & \r
\\ \noalign{\medskip}
  &  &    (u_q,w_q) & \to &    g_{{\mathbf v}}^{\alpha,\beta}(u_q,w_q)= \omega^\alpha_L({\mathbf v})(X^C({\mathbf v}),Y^{V_\beta}({\mathbf v})),
\end{array}\]
where $X,Y$ are vector fields on $Q$ for which $X(q)=u_q$ and $Y(q)=w_q$.

In the natural coordinates  $(q^A,u^A_\alpha)$ on $\tkq$, the coordinate expression of $g_{{\mathbf v}}^{\alpha,\beta}$ is
\[
g_{{\mathbf v}}^{\alpha,\beta} = \frac{\partial^2 L}{\partial u^A_\alpha \partial u^B_\beta}\Big\vert_{ {\mathbf v}} dq^A({\mathbf v}) \otimes dq^B({\mathbf v}).
 \]
 In what follows, we will use the following notations for the coefficients with respect to the basis $\{X_i,{\tilde E}_a\}$ of vector fields on $Q$:
\[
g^{\alpha,\beta}_{ij}({\mathbf v}) =g_{{\mathbf v}}^{\alpha,\beta} (X_i(q),X_j(q)), \quad g^{\alpha,\beta}_{ia}({\mathbf v}) =g_{{\mathbf v}}^{\alpha,\beta} (X_i(q),{\widetilde E}_a(q)), \quad
g^{\alpha,\beta}_{ab}({\mathbf v}) =g_{{\mathbf v}}^{\alpha,\beta} ({\widetilde E}_a(q),{\widetilde E}_b(q)).
\]

Then:
\begin{equation}\label{hess}
g^{\alpha\beta}_{ij} = {X}_i^{V_\alpha} ({X}_j^{V_\beta}(L)),\quad g^{\alpha\beta}_{ia} = {X}_i^{V_\alpha} ({\widetilde E}_b^{V_\beta}(L)), \quad g^{\alpha\beta}_{ab} = {\widetilde E}_a^{V_\alpha} ({\widetilde E}_b^{V_\beta}(L)).
\end{equation}

\begin{definition}
A Lagrangian $L$ is $G$-regular if the matrix $(g^{\alpha\beta}_{ab})$ is non-singular.
\end{definition}
Remark that, in view of Proposition~\ref{rank}, this condition is equivalent with saying that the matrix $\left( \ds \spd{L}{u^a_\alpha}{u^b_\beta}\right)$ is non-singular everywhere.

The maps $g_{{\mathbf v}}^{\alpha,\beta}$ are not completely symmetric (but we have $g_{\mathbf v}^{\alpha,\beta}(u_q,w_q)= g_{{\mathbf v}}^{\beta,\alpha}(w_q,u_q)$). They give rise to the symmetric map
\[ \begin{array}{ccccl}
g_{{\mathbf v}} & : &  (\tkq)_q \times    (\tkq)_q & \to & \r
\\ \noalign{\medskip}
  &  &   ({\mathbf u}= (q;u_{\alpha}) , {\mathbf w}= (q;w_{\beta}))  & \to &    g_{{\mathbf v}}({\mathbf u},{\mathbf w})= g_{{\mathbf v}}^{\alpha,\beta}(u_{\alpha},w_{\beta})
 \end{array}\]
(sum over $\alpha,\beta$). We will next define the mechanical $k$-connection $\Omega^k: T^1_k T^1_k Q \to {\mathfrak g}^k$.

\begin{definition}
 An element  ${\mathbf W}=( W_{1} ,\ldots ,W_{k})\in T^1_k(\tkq)$ such that $\tau^1_{\tkq}({\mathbf W}) ={\mathbf   v}$  is said to be horizontal for the mechanical $k$-connection if it satisfies
\[
g_{{\mathbf v}}(T^1_k \tau^1_Q ({\mathbf W}), (\xi_1,\ldots,\xi_k)_Q(q)) =0,
\]
for all tuples $(\xi_\alpha)\in {\mathfrak g}^k$.
\end{definition}
This is equivalent with
$$\begin{array}{l}
g_{{\mathbf v}}^{\alpha,\beta}\left( T_{ {\mathbf v}}(\tau^1_Q)  (  (W_{\alpha}) ),\xi_{{\beta}_Q}(q)  \right)=0.
\end{array}
$$
Since each element $W_{\alpha}$ can be written in the lifted frame of $\{X_i,{\widetilde E}_a\}$ as
\[
W_{\alpha} = W^i_\alpha X^C_i ({\mathbf v}) + W^a_\alpha {\widetilde E}^C_a ({\mathbf v}) + Z^i_{\alpha\beta} X^{V_\beta}_i ({\mathbf v}) + Z^a_{\alpha\beta} {\widetilde E}^{V_\beta}_a ({\mathbf v})
\]
the condition for ${\mathbf W}$ to be horizontal becomes
\[
g^{\alpha\beta}_{ib}W^i_\alpha + g^{\alpha\beta}_{ab}W^a_\alpha=0.
\]
If we assume that the Lagrangian is $G$-regular, we can conclude that a horizontal ${\mathbf W}=(W_{\alpha})$ takes the form
\[
W_{\alpha} = W^i_\gamma H^\gamma_{i\alpha} ({\mathbf v}) + Z^i_{\alpha\beta} X^{V_\beta}_i ({\mathbf v}) + Z^a_{\alpha\beta} {\widetilde E}^{V_\beta}_a ({\mathbf v}),
\]
where $H^\gamma_{i\alpha} = \delta^\gamma_\alpha X^C_i - {\tilde B}^{\gamma a}_{\alpha i} {\widetilde E}^C_a$ with ${\tilde B}^{\gamma a}_{\alpha i} = g^{b a}_{\beta\alpha} g^{\gamma\beta}_{ib}$.

From this, we can conclude that every element of $T^1_k( T^1_k Q)$ can be written in a `horizontal' and a `vertical part'. Indeed if
\[
W_{\alpha} = W^i_\alpha X^C_i ({\mathbf v}) + W^a_\alpha {\widetilde E}^C_a ({\mathbf v}) + Z^i_{\alpha\beta} X^{V_\beta}_i ({\mathbf v}) + Z^a_{\alpha\beta} {\widetilde E}^{V_\beta}_a ({\mathbf v})
\]
then $W_{\alpha} = \mathbf{HW}_\alpha + \mathbf{VW}_\alpha$, with
\[
\mathbf{HW}_\alpha = W^i_\gamma H^\gamma_{i\alpha} ({\mathbf v})   + Z^i_{\alpha\beta} X^{V_\beta}_i ({\mathbf v}) + Z^a_{\alpha\beta} {\widetilde E}^{V_\beta}_a ({\mathbf v}), \quad \mathbf{VW}_\alpha = (W^a_\alpha  + W^i_\gamma {\tilde B}^{\gamma a}_{\alpha i} ) {\widetilde E}^C_a ({\mathbf v}) .
\]
Remark that the expressions of  $\mathbf{HW}_\alpha$ and $\mathbf{VW}_\alpha$ contain more than just the components of the $\alpha$'th vector $W_{\alpha}$. The mechanical connection is therefore not of simple type.

The corresponding connection map $\Omega^k: T^1_k T^1_k Q \to {\mathfrak g}^k$, is the one that has the property that
\[
\Omega^k (\mathbf{HW}) =0, \quad \Omega^k ((\xi_1,\ldots,\xi_k)_{T^1_kQ}({\mathbf v}) ) = (\xi_1,\ldots,\xi_k).
\]

If we write the {\sc sopde} $k$-vector field ${\mathbf \Gamma}$ in terms of the frame $\{X_i,{\widetilde E}_a\}$ as
\[
\Gamma_\alpha = v^i_\alpha X^C_i   + v^a_\alpha {\widetilde E}^C_a   +
(\widetilde{\Gamma}_\alpha )^j_\beta X_j^{ V_\beta}+ (\widetilde{\Gamma}_\alpha )^a_\beta
\widetilde{E}_a^{ V_\beta},
\]
then
\begin{equation} \label{mg} \mathbf{H\Gamma}_\alpha =  - ( v^i_\gamma {\tilde B}^{\gamma a}_{\alpha i} ) {\widetilde E}_a^C + v^i_\alpha X^C_i    +
(\widetilde{\Gamma}_\alpha )^j_\beta X_j^{ V_\beta}+ (\widetilde{\Gamma}_\alpha)^a_\beta
\widetilde{E}_a^{ V_\beta},\quad
\mathbf{V\Gamma}_\alpha = (v^a_\alpha  + v^i_\gamma {\tilde B}^{\gamma a}_{\alpha i} ) {\widetilde E}^C_a.
\end{equation}

\begin{prop}
The mechanical $k$-connection of an invariant $G$-regular Lagrangian is  a principal $k$-connection on the principal bundle $\pi: T^1_kQ \to (T^1_k Q)/G$.
\end{prop}

\begin{proof}
The condition we need to check is ${\mathcal L}_{\xi_{T^1_kQ}} \omega^k={\mathcal L}_{\xi_{Q}^C} \omega^k = 0$, where $\omega^k$ is the connection (1,1)-$k$-tensor field that is associated to $\Omega^k$, and the Lie derivative is the one we had defined in expression (\ref{lieder}). We first check that the vector fields $H^\gamma_{i\alpha}$ (see above) are invariant vector fields on $T^1_kQ$, i.e.\ $[{\widetilde{E}_a^C},H^\gamma_{i\alpha}]=0$. This will be the case if we can show that
\[
{\widetilde E}^C_a({\tilde B}^{\gamma d}_{\alpha i}) = {\tilde B}^{\gamma b}_{\alpha i} C^d_{ba}.
\]
This relation easily follows because, in view of relation (\ref{hess}) and the invariance of the Lagrangian,  one can show that
\[
{\widetilde E}^C_d( g^{\alpha\beta}_{ab}) = C^e_{db}g^{\alpha\beta}_{ae} + C^e_{da}g^{\alpha\beta}_{eb}, \qquad {\widetilde E}^C_d( g^{\alpha\beta}_{ib}) = C^e_{db}g^{\alpha\beta}_{ie}.
\]
Given that ${\widetilde E}^C_d(g^{\alpha\beta}_{ab}g^{ac}_{\alpha\gamma})={\widetilde E}^C_d(\delta^\beta_\gamma \delta^b_c)=0$ we also obtain
\[
{\widetilde E}^C_d( g^{ec}_{\epsilon\gamma}) = - {\widetilde E}^C_d( g^{\alpha\beta}_{ab}) g^{ac}_{\alpha\gamma} g^{eb}_{ \epsilon\beta}.
\]
for the inverse matrix $g^{ab}_{\alpha\beta}$.  Using these properties and the expression ${\tilde B}^{\gamma d}_{\alpha i}=g^{bd}_{\beta\alpha}g^{\gamma \beta}_{ib}$  we obtain the desired result.

Assume now that ${\mathbf H}$ is a horizontal $k$-vector field on $T^1_kQ$. Then
 $({\mathcal L}_{{\widetilde E}_a^C} \omega^k) ({\mathbf H})=-\omega^k({\mathcal L}_{\widetilde{E}_a^C}{\mathbf H} )$. If we set $H_\alpha=W^i_\gamma H^\gamma_{i\alpha}   + Z^i_{\alpha\beta} X^{V_\beta}_i + Z^b_{\alpha\beta} {\widetilde E}^{V_\beta}_b$, we easily see that
\[
({\mathcal L}_{\widetilde{E}_a^C}{\mathbf H} )_\alpha = [{\widetilde{E}_a^C},H_\alpha] = {\widetilde{E}_a^C}(W^i_\gamma) H^\gamma_{i\alpha}   + {\widetilde{E}_a^C}(Z^i_{\alpha\beta}) X^{V_\beta}_i + {\widetilde{E}_a^C}(Z^b_{\alpha\beta}) {\widetilde E}^{V_\beta}_b - Z^b_{\alpha\beta} C_{ab}^d {\widetilde E}^{V_\beta}_d,
\]
which are the components of again a horizontal $k$-vector field. When $\omega^k$ is applied to it, we will get zero and thus is $({\mathcal L}_{{\widetilde E}_a^C} \omega^k) ({\mathbf H})=0$. With the same reasoning one may show that
 $({\mathcal L}_{{\widetilde E}_a^C} \omega^k) ({\mathbf V})=0$ for all vertical $k$-vector fields ${\mathbf V}$ on $T^1_kQ$.
\end{proof}

 Since ${\mathbf \Gamma}$ is $G$-invariant, and since the mechanical connection is principal, the horizontal component  $\mathbf{H\Gamma}$ of ${\mathbf \Gamma}$ is the horizontal lift $\breve{\mathbf \Gamma}^H$ of the reduced $k$-vector field $\breve{\mathbf \Gamma}$. By definition the horizontal lift of an integral section $(q^i=\phi^i(t), v^i_\alpha =\phi^i_\alpha(t), w^a_\alpha = \phi^a_\alpha(t))$ of $\breve{\mathbf \Gamma}$ is an integral section of $\breve{\mathbf \Gamma}^H =\mathbf{H\Gamma}_\alpha$. In principle, we need to rewrite $\mathbf{H\Gamma}_\alpha$ in terms of the frame $\{Z_A\} = \{{\widehat E}_a,X_i\}$, and use expressions (\ref{intsectionframe}) to calculate an integral section $(q^i=\phi^i(t), q^a=\phi_H^a(t), v^i_\alpha =\phi^i_\alpha(t), w^a_\alpha = \phi^a_\alpha(t))$ (in quasi-velocities) of $\mathbf{H\Gamma}$. However, we only require the equations from which we may determine $\phi_H^a(t)$, since the remainder $(q^i=\phi^i(t), v^i_\alpha =\phi^i_\alpha(t), w^a_\alpha = \phi^a_\alpha(t))$ is determined by the reduced $k$-vector field $\breve{\mathbf \Gamma}$. In view of the first relations in (\ref{intsectionframe}) the equations for $\phi^a_H(t)$ are given by
 \begin{equation} \label{horliftmech}
\fpd{\phi^a_H}{t^\alpha} = -\phi^i_\gamma \Big(K^a_b (\gamma^c_i A^b_c \delta^\gamma_\alpha +  {\tilde B}^{\gamma b}_{\alpha i}) \circ {\breve\phi}_H\Big),
\end{equation}
where we have made use of the expressions $X_i = \partial/\partial q^i - \gamma^a_i {\widehat E}_a$ and ${\widetilde E}_b =K^a_b\partial/\partial q^a$.

When we use the mechanical $k$-connection, the reconstruction equation (\ref{receq}) becomes, in view of expression (\ref{mg}),
 \begin{equation} \label{mreceq}
(g^{-1} g^{(1)})_\alpha =   \Big((v^a_\alpha  + v^i_\gamma {\tilde B}^{\gamma a}_{\alpha i} )  \circ \breve{\phi}_H \Big)\, E_a.
 \end{equation}

When we put everything together, we get:
\begin{prop} Let $L$ be a regular, $G$-regular, invariant Lagrangian. In order to carry out the
reconstruction by means of the mechanical connection, one needs to solve successively
\begin{enumerate}
\item the Lagrange-Poincar\'e field equations (\ref{l-eq}) for ${\breve\phi}(t)=(q^i=\phi^i(t), v^i_\alpha =\phi^i_\alpha(t), w^a_\alpha = \phi^a_\alpha(t) )$.
\item the equations (\ref{horliftmech}) for $\phi^a_H(t)$.
\item the reconstruction equation (\ref{mreceq}) for $g(t)$,
\end{enumerate}
to obtain the solution $\phi(t) =g(t){\breve\phi}_H(t)$ of the Euler-Lagrange field equations (\ref{lfield}).
\end{prop}


\section{An application on harmonic maps}

Harmonic maps are smooth maps $\phi: M \to   Q$ between two Riemannian manifolds $(M,g)$ and $(Q,h)$ which have the property that their tension field, given by
\[
\tau(\phi) = g^{\alpha\beta} \left(\spd{\phi^A}{t^\alpha}{t^\beta} -{}^g\Gamma^\delta_{\alpha\beta} \fpd{\phi^A}{t^\delta} + {}^h\Gamma^A_{BC} \fpd{\phi^B}{t^\alpha}\fpd{\phi^C}{t^\beta}\right),
\]
vanishes (see e.g.\ \cite{HeleinWood}). Here ${}^g\Gamma^\delta_{\alpha\beta}$ and ${}^h\Gamma^A_{BC}$ stand for the Christoffel symbols of $g$ and $h$, respectively. In the special case where $(M,g)$ is just $\r^k$ with its standard Euclidean metric, it is well-known that the above conditions can be thought of as the Lagrangian field equations of the Lagrangian
\[
L: T^1_k Q \to \r,\,\, (q^A, u^A_\alpha) \mapsto \frac{1}{2} \delta^{\alpha\beta}h_{AB}(q)u^A_\alpha u^B_\beta.
\]
For this Lagrangian, one may check that the $k$-vector field ${\mathbf\Gamma}$ with
\[
\Gamma_\alpha = u^A_\alpha\fpd{}{q^A} - \Gamma^A_{BC}u^B_\alpha u^C_\beta \fpd{}{u^A_\beta}
\]
is Lagrangian (we will simply write ${}^h\Gamma^A_{BC}=\Gamma^A_{BC}$ from now on).

Let's assume that the metric $h$ has a symmetry Lie group $G$  which acts freely and properly to the left as isometries, and that the corresponding basis of invariant vertical vector fields is denoted by ${\widehat E}_a$, as before. We may define a principal connection on  $Q\to Q/G$ by declaring that horizontal vector fields lie in the complement of vertical vector fields. This is equivalent with saying that the vector fields $X_i$ on $Q$ are defined by the relations $h(X_i,{\widehat E}_a) =0$ and by the fact that they project on coordinate vector fields on $Q/G$ (this is, in fact, the definition of the 'mechanical' connection of the Riemannian metric $h$, see e.g.\ \cite{MC}). We will set $h_{ij}=h(X_i,X_j)$  and $h_{ab}=h({\widehat E}_a,{\widehat E}_b)$. These are all invariant functions. We will further assume that the vertical part of the metric, $h_{ab}$, comes from a bi-invariant metric on $G$, or, equivalently, from an $Ad$-invariant inner product on $\g$. That is, we will assume that $h_{ab}$ are all constants satisfying
\[
h_{ab}C^{b}_{cd} + h_{cb}C^{b}_{ad}=0.
\]
In view of $\Upsilon^b_{ia} =-\gamma^c_i C^b_{ca}$, we also obtain that
\[
h_{ab}\Upsilon^{b}_{ic} + h_{cb}\Upsilon^{b}_{ia}=0.
 \]
 From these relations,  we may also see that $\delta^{\alpha\beta}  h_{db} C^b_{ac} w^c_\beta w^d_\alpha =0$ and $\delta^{\alpha\beta} h_{ab}\Upsilon^b_{ic} w^c_\beta w^a_\alpha =0$.

The reduced Lagrangian is $l=\frac{1}{2} \delta^{\alpha\beta}(h_{ij}v^i_\alpha v^j_\beta+h_{ab}w^a_\alpha w^b_\beta)$.
If one takes these last  two properties into account in the calculation of the the Lagrange-Poincar\'e equations (\ref{eqelred}), one easily verifies that the $k$-vector field  ${\breve\Gamma}_\alpha$ of Lemma~\ref{ivfgam}, with
\begin{equation} \label{ex1}
({\breve\Gamma}_\alpha)^j_\beta  =  -  \Gamma^j_{kl}v^k_\alpha v^l_\beta +h^{ji}h_{ab}K^b_{ik}v^k_\beta w^a_\alpha ,\qquad ({\breve\Gamma}_\alpha)^b_\beta = - \Upsilon^b_{kd}v^k_\beta w^d_\alpha
\end{equation}
satisfies the equations (\ref{eqelred}). The functions $\Gamma^i_{jk}$ are the Christoffel symbols of $h_{ij}$ (which is a Riemannian metric on $Q/G$).

The integral  sections of the reduced $k$-vector field $\breve{\mathbf \Gamma}$ will therefore be solutions of the PDEs
\[
\fpd{\phi^j}{t^\alpha} = \phi^j_\alpha, \qquad \fpd{\phi^j_\alpha}{t^\beta}  =  -  \Gamma^j_{kl}\phi^k_\alpha \phi^l_\beta +h^{ji}h_{ab}K^b_{ik}\phi^k_\beta \phi^a_\alpha ,\qquad \fpd{\phi^b_\alpha}{t^\beta}  = - \Upsilon^b_{kd}\phi^k_\beta \phi^d_\alpha.
\]
 From the first two equations, it is clear that the curvature $K^b_{ik}$ of the connection acts as an obstruction for the reduced equation to be again of the type of a harmonic map  $(\r^k,\delta_{\alpha\beta}) \to (Q/G,h_{ij})$.

In order to reconstruct the integral section of the field equations, we need to compute the horizontal lift ${\breve\phi}_H$ of an integral section of ${\breve{\mathbf\Gamma}}$, with respect to the mechanical $k$-connection we had introduced in Section~\ref{sec63}. This connection takes a rather simple form here. Indeed, it is clear that in the current setting, where we have defined the connection on $Q \to Q/G$ as the one for which  $h_{ia} =0$, we have that $g^{\alpha\beta}_{ij} = \delta^{\alpha\beta}h_{ij}$, $g^{\alpha\beta}_{ia} = 0$, $g^{\alpha\beta}_{ab} = \delta^{\alpha\beta}h_{ab}$ and therefore also ${\tilde B}^{\gamma a}_{\alpha i} =0$. The equation (\ref{horliftmech}) from which we may determine the horizontal lift takes therefore the form
 \begin{equation} \label{horliftmech2}
\fpd{\phi^a_H}{t^\alpha} = -\phi^i_\alpha \Big(\gamma^c_i K^a_b  A^b_c \circ {\breve\phi}_H\Big).
\end{equation}
Likewise, the reconstruction equation (\ref{mreceq}) becomes (with $v^a_\alpha = A^a_b w^b_\alpha$):
 \begin{equation} \label{mreceq2}
(g^{-1} g^{(1)})_\alpha =   \Big(A^a_b \circ \breve{\phi}_H\Big) \phi^b_\alpha E_a.
 \end{equation}

We will use an explicit example to show how one may reconstruct a solution, from a solution of the Lagrange-Poincar\'e equations. We will consider a 4-dimensional matrix Lie group $G$, whose typical element $g=(x,y,z,\theta)$ is of the type
\[
\left[ \begin {array}{cccc} 1&y\cos \theta +x\sin \theta &-y\sin \theta +x\cos \theta &z
\\ 0&\cos \theta &-\sin\theta &
x\\0&\sin \theta &\cos \theta &
-y\\ 0&0&0&1\end {array} \right].
\]
Left multiplication $L_g: G \to G$ is then given by
\begin{equation} \label{leftmult}
(\bar x, \bar y, \bar z, \bar\theta)\mapsto (x+ \bar x \cos\theta+\bar y \sin\theta,y - \bar x \sin\theta+\bar y \cos\theta, z + \bar z+ (x\bar x + y\bar y)\sin\theta +  (y\bar x-x\bar y)  \cos\theta,\theta+\bar\theta).
\end{equation}
In \cite{Ghanam} it has been shown that this is a group representation of the Lie algebra whose only non-vanishing brackets are given by $[e_2,e_3]=e_1$, $[e_2,e_4]=-e_3$ and $[e_3,e_4]=e_2$. One may find in \cite{Ghanam} the following basis for right-invariant vector fields
\[
{\widetilde E}_x = \fpd{}{x}-y\fpd{}{z}, \quad {\widetilde E}_y = \fpd{}{y}+x\fpd{}{z}, \quad {\widetilde E}_z = \fpd{}{z}, \quad {\widetilde E}_\theta = \fpd{}{\theta}-x\fpd{}{y} +y\fpd{}{x},
\]
or, if we set ${\widetilde E}_a = K^b_a \partial/\partial q^b$, then
\[K=
\left[ \begin {array}{cccc}  1&0  &-y  &0 \\
  0 & 1  & x  &0 \\
0    & 0 & 1 & 0\\
 y    & -x  &0  & 1
     \end {array} \right].
\]
One may easily verify that the list below gives a basis, consisting only of left-invariant vector fields:
\begin{eqnarray*}
&& {\widehat E}_x =  \cos\theta\left(\fpd{}{x} + y\fpd{}{z} \right)-\sin\theta \left( \fpd{}{y} -  x\fpd{}{z}\right),\quad {\widehat E}_z = \fpd{}{z},\\&&  {\widehat E}_y = \sin\theta\left(\fpd{}{x} + y\fpd{}{z} \right) +\cos\theta \left( \fpd{}{y} - x\fpd{}{z}\right),   \quad {\widehat E}_\theta = \fpd{}{\theta}.
\end{eqnarray*}
With these vector fields, the only non-vanishing structure constants are $C_{xy}^z=-2$, $C_{x\theta}^y=1$ and $C^x_{y\theta} =-1$. The matrix $A^a_b$ in the expression ${\widehat E}_a = A_a^b {\widetilde E}_b$ is then
\[
A= \left[ \begin {array}{cccc} \cos\theta& -\sin\theta  &2(y\cos \theta +x\sin \theta) &0
\\ \sin\theta&\cos \theta &  2(y\sin \theta -x\cos \theta) &
0\\ 0& 0 &1 &
0\\ -y&x&x^2+y^2&1\end {array} \right].
\]
Remark, for later use, that it is independent of $z$.

We will consider the manifold  $Q=\r\times G$ with its natural $G$-action. We will denote the coordinate  on $Q/G=\r$ by $q$, and $(x,y,z,\theta)$ for those on $G$, as before. The Riemannian metric
\[
h= dq \odot dq+ \gamma dq \odot  d\theta +dx\odot  dx+ dy \odot  dy- y dx\odot d\theta+ x dy\odot d\theta+ dz\odot d\theta
\]
satisfies ${\mathcal L}_{{\widetilde E}_a}h=0$, so that it is an invariant metric. The corresponding principal connection on $Q\to Q/G$ can be represented by the unique horizontal vector field  $X = \partial/\partial q -\gamma \partial/\partial z$ which projects on $\partial/\partial q$. Therefore, all $\Upsilon^b_{qa}=-\gamma C_{za}^b =0$.

 In the notation of what preceded, we have $h_{qq}=h(X,X)=1$. The vertical part of the metric,
\[
(h_{ab}) = dx\odot  dx+ dy \odot  dy- y dx\odot d\theta+ x dy\odot d\theta+ dz\odot d\theta
\]
represents (as it was already  mentioned in \cite{Ghanam}) a bi-invariant metric on $G$. We are therefore in the situation of the previous paragraph. We can use the reduced Lagrangian $k$-vector field (\ref{ex1}) to compute  integral sections $(t^\alpha) \mapsto (\phi^q(t), v^q_\alpha(t), w^x_\alpha(t),w^y_\alpha(t),w^z_\alpha(t),w^\theta_\alpha(t))$ of the Lagrange-Poincar\'e field equations. They satisfy:
\[
\fpd{\phi^q}{t^\alpha} = v^q_\alpha(t),\qquad \fpd{v^q_\beta}{t^\alpha}  = 0,\quad \fpd{w^a_\beta}{t^\alpha}  = 0,\]
from which we may conclude that
\[
\phi^q(t) = c_{\alpha}^q t^\alpha + b^q, \qquad w^a_\beta(t) = c^a_{\beta}.
\]

The equations (\ref{horliftmech2}) for the horizontal lifts are now
\[
\fpd{\phi^x_H}{t^\alpha} = 0, \qquad \fpd{\phi^y_H}{t^\alpha} = 0,\qquad \fpd{\phi^z_H}{t^\alpha} = -\gamma  c^q_\alpha, \qquad \fpd{\phi^\theta_H}{t^\alpha} = 0.
\]
It follows that:
\[
\phi^z_H(t) = - \gamma c_\alpha^q t^\alpha +b^z, \qquad \phi^a_H(t)= b^a \quad(a\neq z).
\]
Since the matrix $A$ does not depend on $z$,  the right-hand side of the reconstruction equations (\ref{mreceq2}) contains only the constants $c^a_\beta$ and $b^a$. It is therefore of the form
\[
C_\alpha^x E_x + C_\alpha^y E_y + C_\alpha^z E_z + C_\alpha^\theta E_\theta,
\]
for some other constants $C_\alpha^a$ (with, in particular, $C^\theta_\alpha =c^\theta_\alpha$).

With the help of the map (\ref{leftmult}), the expression of $g^{-1}g^{(1)}$, with $g(t)=(\phi_g^x(t),\phi_g^y(t),\phi_g^z(t),\phi_g^\theta(t))$ can be computed to be
\begin{eqnarray*}
(g^{-1}g^{(1)})_\alpha &=& \left( \cos(\phi_g^\theta(t)) \fpd{\phi_g^x}{t^\alpha} - \sin(\phi_g^\theta(t)) \fpd{\phi_g^y}{t^\alpha} \right)E_x + \left( \sin(\phi_g^\theta(t)) \fpd{\phi_g^x}{t^\alpha} + \cos(\phi^\theta_g(t)) \fpd{\phi_g^y}{t^\alpha} \right)E_y
\\
&& + \left(   \phi_g^x(t)\fpd{\phi_g^y}{t^\alpha}  - \phi_g^y(t) \fpd{\phi_g^x}{t^\alpha}  +\fpd{\phi_g^z}{t^\alpha} \right)E_z + \fpd{\phi_g^\theta}{t^\alpha} E_\theta.
\end{eqnarray*}

From the first two reconstruction equations (\ref{mreceq2}) we may then conclude that
\[ \fpd{\phi_g^x}{t^\alpha} = C^x_\alpha \cos(\phi_g^\theta(t)) +C^y_\alpha\sin(\phi_g^\theta(t)), \qquad \fpd{\phi_g^y}{t^\alpha} =- C^x_\alpha \sin(\phi_g^\theta(t)) +C^y_\alpha\cos(\phi_g^\theta(t)).
\]
The last reconstruction equation leads to  $\phi_g^\theta(t) = c^\theta_\alpha t^\alpha + B^\theta$. For computational convenience, let's consider only the simple case where the solution for $\phi^\theta_g$ is given by
\[
\phi^\theta_g(t) = t^1.
\]
 Due to the assumed integrability, the second partial derivatives $\displaystyle\fpd{}{t^1}\left(\fpd{\phi^x_g}{t^\beta}\right)$ and $\displaystyle\fpd{}{t^\beta}\left(\fpd{\phi_g^x}{t^1}\right)$ should agree. Since the last derivative automatically vanishes, we may conclude that the constants $C^x_\alpha$ are zero when $\alpha>1$. Then:
\[
\phi^x_g(t) = C^x_1\sin t^1 - C^y_1\cos t^1+ B^x.
\] Likewise,
\[
\phi^y_g(t) = C^x_1\cos t^1 +C^y_1\sin t^1+B^y.
\]
With that, the solution of (\ref{mreceq2}) for $\phi^z_g$ is
\[
\phi^z_g(t) = -(B^x C^x_1 + B^yC^y_1) \cos t^1- (B^xC^y_1 -B^yC^x_1) \sin t^1 +((C^x_1)^2+(C^y_1)^2)t^1+ C^z_\alpha t^\alpha + B^z.
\]

If we use the left multiplication (\ref{leftmult}), one may easiily see that the solution $\phi(t)=g(t)\phi_H(t)$ of the Lagrangian field equations can be written as:
\[
\begin{array}{ll}
   \phi^q(t) = c_{\alpha}^q t^\alpha + b^q, &\qquad \phi^\theta(t)=t^1+b^\theta,
\\[1mm]
  \phi^x(t) = {\bar C}^x_1\sin t^1 - {\bar C}^y_1\cos t^1 +{\bar B}^x, &\qquad \phi^y(t) = {\bar C}^x_1\cos t^1  + {\bar C}^y_1\sin t^1 +{\bar B}^y,\\[1mm]
&\hspace*{-6.1cm}
 \phi^z(t) = -({\bar B}^x {\bar C}^x_1 + {\bar B}^y{\bar C}^y_1) \cos t^1 - ({\bar B}^x{\bar C}^y_1 -{\bar B}^y{\bar C}^x_1) \sin t^1  + {\bar C}^z_\alpha t^\alpha + {\bar B}^z .
\end{array}
\]


{\sc Acknowledgements.} L. B\'{u}a and M.\ Salgado  acknowledge  the financial support of the Ministerio de Ciencia e Innovaci\'{o}n (Spain), projects MTM2011-22585, MTM2011-15725-E, Ministerio de Economía y Competitividad MTM2014-54855-P.


\begin{thebibliography}{MM}



\itemsep 1pt plus 1pt

\bibitem{AM}
R.A. Abraham, J.E. Marsden. Foundations of Mechanics,
(Second Edition),  Benjamin-Cummings Publishing Company, New York, (1978).

\bibitem{BBS}
L.\ B\'{u}a, I.\ Bucataru and M.\ Salgado, Symmetries, Newtonoid vector fields and conservation laws in the Lagrangian k-symplectic formalism. Rev.\ Math.\ Phys.\ 24 (2012), no. 10, 1250030, 24 pp.

\bibitem{CGR} M.\ Castrill\'on L\'opez, P.L.\ Garc\'ia and C.\ Rodrigo, Euler-Poincar\'e reduction in principal bundles by a subgroup of the structure group, J.\ Geom.\ Phys. 74 (2013) 352--369.

\bibitem{CRS} M.\ Castrill\'on L\'opez, T.S.\ Ratiu and S.\ Shkoller, Reduction in principal fibre bundles: Covariant Euler-Poincar\'e equations, Proc.\ Amer.\ Math.\ Soc.\ 128 (2000) 2155--2164.

\bibitem{CCI91} J.F.\ Carinena, M.\ Crampin and L.A.\ Ibort, On the multisymplectic formalism for first order field theories, Differential Geometry and its Applications  1 (1991) 345--374.


\bibitem{Cendra} H.\ Cendra, J.E.\ Marsden and T.S.\ Ratiu,  Lagrangian reduction by stages. Mem.\ Amer.\ Math.\ Soc.\ 152 (2001), no.\ 722.

\bibitem{bar1} A.\ Echeverr\'\i a-Enr\'\i quez, M.C.\  Mu\~noz-Lecanda and N.\ Rom\'an-Roy,  Geometry of Lagrangian first-order classical field theories, Forts.\ Phys.\ 44  (1996) 235--280.



\bibitem{RRA} M.\ Crampin and T.\ Mestdag,  Reduction and reconstruction aspects of second-order dynamical systems with symmetry. Acta Appl. Math. 105 (2009) 241--266.

\bibitem{CP} M.\ Crampin and F.A.E.\ Pirani, Applicable differential geometry. London Mathematical Society Lecture Note Series, 59. Cambridge University Press, Cambridge, 1986.

\bibitem{LR} M.\ de Leon and P.R.\ Rodrigues, Methods of differential geometry in analytical mechanics. North-Holland Mathematics Studies, 158. North-Holland Publishing Co., Amsterdam, 1989.

\bibitem{Ellis} D.C.P.\ Ellis, F.\ Gay-Balmaz, D.D.\ Holm and T.S.\ Ratiu, Lagrange-Poincar\'e field equations, J.\ Geom.\ Phys.\   61  (2011) 2120--2146.

\bibitem{Ellis2}  D.C.P.\ Ellis, F.\ Gay-Balmaz, D.D.\ Holm, V.\ Putkaradze, and T.S.\ Ratiu,  Symmetry Reduced Dynamics of Charged Molecular Strands, Arch.\ Ration.\ Mech.\ Anal.\ 197 (2010) 811--902.

\bibitem{Ghanam} R.\ Ghanam, G.\ Thompson and E.J.\ Miller, Variationality of four-dimensional Lie group Connections, Journal of Lie Theory 14 (2004) 395--425.

\bibitem{GIM1} M.J.\ Gotay, J.\ Isenberg, J.E.\ Marsden and R.\ Montgomery, Momentum Maps and Classical Relativistic Fields. Part I: Covariant Field Theory,  arXiv:physics/9801019v2(2004).


\bibitem{gunther}
    C.\ G\"unther, The polysymplectic Hamiltonian formalism in field
    theory and calculus of variations I: The local case,  J.\ Differential Geom.\ 25 (1987) 23--53.



\bibitem{HeleinWood} F.\ H\'elein and J.C.\ Wood, Harmonic maps, In: D.\ Krupka and D.J.\ Saunders, Handbook of Global Analysis, Elsevier (2008).

\bibitem{Kana}  I.V.\ Kanatchikov, Canonical structure of classical field theory in the polymomentum phase space,  Rep.\ Math.\ Phys.\  41  (1998) 49--90.

\bibitem{KijTul} J.\  Kijowski and W.M.\  Tulczyjew, A symplectic framework for field theories. Lecture Notes in Physics, 107. Springer-Verlag, New York, 1979.

\bibitem{Krupka} D.\ Krupka, Lagrange theory in fibered manifolds, Rep.\ Math.\ Phys.\ 2 (1971) 121--133.


\bibitem{mod2}   M.\ de Le\'{o}n, E.\ Merino  and M.\  Salgado, $k$-cosymplectic manifolds and Lagrangian field theories, J.\ Math.\ Phys.\ 42  (2001) 2092--2104.

\bibitem{marrero} J.C.\ Marrero, N.\ Rom\'an-Roy, M.\ Salgado and S.\ Vilari\~no, Reduction of polysymplectic manifolds, J.\ Phys.\ A: Math.\ Theor.\ 48 (2015) 055206 (43pp).

\bibitem{Mestdag} T.\ Mestdag, A Lie algebroid approach to Lagrangian systems with symmetry, In: J. Bures et al (eds.), Differential Geometry and its Applications, Proc. Conf., Prague (Czech Republic) (2005), 523--535.

\bibitem{MC}
T.\ Mestdag and M.\ Crampin, Invariant Lagrangians, mechanical connections and the Lagrange-Poincar\'{e} equations, J.\ Phys.\ A: Math.\ Theor. \ 41 (2008) 344015 (20pp).


\bibitem{no1}  M.\ McLean and L.K.\  Norris, Covariant field theory on frame bundles of fibered manifolds, J.\ Math.\ Phys.\  41  (2000)  6808--6823.

\bibitem{fam} F.\ Munteanu, A.M.\ Rey and M.\ Salgado, The G\"{u}nther's formalism in classical field theory: momentum map and reduction,
 J.\ Math.\ Phys.\ 45 (2004) 1730--1751.

\bibitem{rsv07} N.\ Rom\'an-Roy, M.\ Salgado and S.\  Vilari\~no, Symmetries and Conservation Laws in G\"unter k-symplectic formalism of Field Theory,
Reviews in Mathematical Physics, 19 (2007), 1117--1147.

\bibitem{RRSV} N.\ Rom\'an-Roy, M.\ Salgado, and S.\ Vilari\~no, SOPDEs and nonlinear connections, Publ.\ Math.\ (Debrecen) 78 (2011),  297--316.

\bibitem{sarda1} G.\ Sardanashvily, Generalized Hamiltonian Formalism for Field Theory. Constraint Systems. World Scientific, Singapore (1995).

\bibitem{Saunders} D.J.\ Saunders, The geometry of jet bundles, Cambridge University Press (1989).

\bibitem{JV} J.\ Vankerschaver, Euler-Poincar\'e reduction for discrete field theories, J.\ Math.\ Phys.\   48  (2007) 032902.

 \end{thebibliography}
\end{document}